\documentclass[12pt]{article}
\usepackage{amsmath, amsthm, amsfonts, amssymb}
\usepackage[all,cmtip]{xy}
\usepackage{eucal}
\usepackage{graphicx, color}
\usepackage[toc,page]{appendix}

\newtheorem{theorem}{Theorem}
\newtheorem{lemma}[theorem]{Lemma}
\newtheorem{proposition}[theorem]{Proposition}

\theoremstyle{definition}
\newtheorem{definition}[theorem]{Definition}

\theoremstyle{remark}
\newtheorem{remark}[theorem]{Remark}

\numberwithin{equation}{section} \numberwithin{theorem}{section}
\newcommand\lbb[1]{\label{#1}}
\def\<{\langle}
\def\>{\rangle}

\def \p{\partial}
\def \pa{\partial}
\def \cp{\mathbb C[\partial]}

\newcommand{\CC}{\mathbb{C}}
\newcommand{\NN}{\mathbb{N}}
\newcommand{\QQ}{\mathbb{Q}}
\newcommand{\ZZ}{\mathbb{Z}}

\newcommand{\fg}{\mathfrak{g}}
\newcommand{\fh}{\mathfrak{h}}

\newcommand{\fc}{\mathfrak{c}}
\newcommand{\fl}{\mathfrak{l}}
\newcommand{\fs}{\mathfrak{s}}
\def\al{\alpha}                         %%% some abbreviations
\def\be{\beta}

\def\ep{\varepsilon}
\def\la{\lambda}
\def\La{\Lambda}

\def\so{{\mathfrak{so}}}
\def\cso{{\mathfrak{cso}}}
\def\fo{{\mathfrak{o}}}

\def\A{{\mathcal{A}}}

\DeclareMathOperator{\Ind}{Ind} 
\DeclareMathOperator{\Ir}{Ir}
 
\DeclareMathOperator{\Sing}{Sing}

 \DeclareMathOperator{\Lie}{Lie}
 \DeclareMathOperator{\Cur}{Cur}

 \DeclareMathOperator{\Vir}{Vir}
\newcommand{\Alg}{{\rm Alg}}

\begin{document}

\title{Classification of finite irreducible modules
over the Lie conformal superalgebra $CK_6$}

\author{Carina Boyallian$^\dagger$, Victor G.~Kac \thanks{{Department of Mathematics, MIT, Cambridge, MA 02139, USA}
 - {kac@math.mit.edu}} and Jos\'e I. Liberati \thanks{ {Famaf-Ciem, Univ. Nac. C\'ordoba,
Ciudad Universitaria,    (5000) C\'ordoba, Argentina}
 - \hfill \newline {boyallia@mate.uncor.edu}, {joseliberati@gmail.com}}
}

\maketitle

\begin{abstract}
We classify all continuous degenerate irreducible modules over the exceptional
linearly compact Lie superalgebra $E(1,6)$, and all finite degenerate 
irreducible modules over the exceptional Lie conformal
superalgebra $CK_6$, for which $E(1,6)$ is the annihilation algebra.
\end{abstract}

%\vfill
%\pagebreak

%%%%%%%%%%%%%%%%%%%%%%%%%%%%%%%%%%%%%%%%%%%%%%%%%%%%%%%%%%%%%%%%%%%%%%%%%
%
%%%%%%%%%%%%%%%%%%%%%%%%%%%%%%%%%%%%%%%%
\section{Introduction}\lbb{sintro}
%%%%%%%%%%%%%%%%%%%%%%%%%%%%%%%%%%%%%%%%

\
Lie conformal superalgebras encode the singular part of the
operator product expansion of chiral fields in two-dimensional
quantum field theory \cite{K1}.

A complete classification of finite simple Lie conformal
superalgebras was obtained in \cite{FK}.  The list consists of current
Lie conformal superalgebras $\Cur \fg$, where $\fg$ is a simple
finite-dimensional Lie superalgebra, four series of ``Virasoro
like'' Lie conformal superalgebras $W_n (n \geq 0)$, $S_{n,b}$
and $\tilde{S}_n (n \geq 2, b \in \CC)$, $K_n (n\geq 0,\, n \neq
4)$, $K'_4$, and the exceptional Lie conformal superalgebra $CK_6$.

All finite irreducible representations of the simple Lie conformal
superalgebras $\Cur \fg$, $K_0=\Vir$ and $K_1$ were constructed
in \cite{CK}, and those of $S_{2,0}, W_1=K_2$, $K_3$, and $K_4$ in \cite{CL}.
More recently, the problem has been solved for all Lie conformal
superalgebras of the three series $W_n$, $S_{n,b}$, and
$\tilde{S}_n $ in \cite{BKLR}, and for all Lie conformal superalgebras of
the remaining series $K_n (n \geq 4)$ in \cite{BKL}.  The construction
in all cases relies on the observation that the representation
theory of a Lie conformal superalgebra $R$ is controlled by the
representation theory of the associated (extended) {\em
  annihilation algebra} $\fg = (\Lie R)_+$ \cite{CK}, thereby reducing
the problem to the construction of continuous irreducible modules
with discrete topology over the linearly compact Lie superalgebra
$\fg$.

The construction of the latter modules consists of two parts.
First one constructs a collection of continuous $\fg$-modules
$\Ind (F)$, associated to all finite-dimensional irreducible
$\fg_0$-modules $F$, where $\fg_0$ is a certain subalgebra of
$\fg (=\fg\fl (1|n) )$ or $\fs\fl (1|n)$ for the $W$ and $S$
series, and $=\fc\fs\fo_n$ for the $K_n$ series and $CK_6$).

The irreducible $\fg$-modules $\Ind (F)$ are called
non-degenerate.  The second part of the problem consists of two
parts:  (A) classify the $\fg_0$-modules $F$, for which the
$\fg$-modules $\Ind (F)$ are non-degenerate, and (B) construct
explicitly the irreducible quotients of $\Ind (F)$, called {\em
  degenerate} $\fg$-modules, for reducible $\Ind (F)$.

Both problems have been solved for types $W$ and $S$ in \cite{BKLR}, and
it turned out, remarkably, that all degenerate modules occur as
cokernels of the super de~Rham complex or their duals.  More
recently both problems have been solved for type $K$ in \cite{BKL}, and
it turned out that again, all degenerate modules occur as
cokernels of a certain complex or their duals. This complex
is a certain reduction of the super de~Rham complex, called in \cite{BKL}
the super contact complex (since it is a ``super'' generalization
of the contact complex of M.~Rumin).

The present paper is the first in the series of three papers on
construction of all finite irreducible representations of the Lie
conformal superalgebra $CK_6$.  In this paper we find all singular
vectors of the $\fg$- modules $\Ind (F)$ for $\fg = E (1,6)$, where $F$
is a finite-dimensional irreducible representation of the Lie
algebra $\fc\fs\fo_6$.  In particular, we find the list of all finite
degenerate irreducible modules over $CK_6$.  In our second paper
we give a proof of the key Lemma \ref{lem:deg}, and in the third
paper construct the complexes, consisting of all degenerate $E
(1,6)$-modules $\Ind (F)$, providing thereby an explicit
construction of all finite irreducible degenerate modules over
$CK_6$.

All degenerate $E(1,6)$-modules $\Ind (F)$ can be represented by
the diagram below (very similar to that for $E(5,10)$ in \cite{KR2}),
with the point $(4,0)$ excluded, where the nodes represent the highest
weights of the modules $\Ind (F)$, and arrows represent the morphisms
between these modules.  Here $\lambda_2,\lambda_1,\lambda_3$ are the fundamental weights
of $\fs\fo_6 =A_3$ (where $\lambda_1$ is attached to the middle node
of the Dynkin diagram).

% \begin{figure}[htbf]
% \begin{center}
% \includegraphics[width=16cm]{CK6-texfigfile.ps}
% \end{center}
% %\caption{Fig: 2}
% \end{figure}

In the subsequent publication we shall compute cohomology of these
complexes, providing thereby an explicit construction of all
degenerate continuous irreducible $E(1,6)$-modules, hence of all degenerate
finite irreducible $CK_6$-modules.

This work is organized as follows: In section \ref{sec:formal} we
introduce notations and definitions of formal distributions, Lie
conformal superalgebras and their modules. In section \ref{sK}, we
introduce the Lie conformal algebra $CK_6$, the annihilation Lie
algebra $E(1,6)$ and the induced modules. In section
\ref{ssingular}, we classify the singular vectors of the induced
modules (Theorem \ref{sing-vect}), and in Theorem
\ref{degenerate} we present the list of highest weights of
degenerate irreducible modules. The last part of this section and Appendix \ref{appendix A}  are
devoted to their proofs through several lemmas. More precisely, we used the software Macaulay2 to simplify the computations, and Appendix
\ref{appendix A} contains the notations in the complementary files that use
Macaulay2 and the reduction procedure to find simplified
conditions on singular vectors of small degree. All these
simplified and equivalent list of equations, obtained with the software as explained in Appendix \ref{appendix A}, are analyzed in details in the proofs
of Lemmas \ref{m5}-\ref{m1} in Section~\ref{ssingular}.

% \ 

\begin{figure}[htbf]
\begin{center}
\includegraphics[width=14cm]{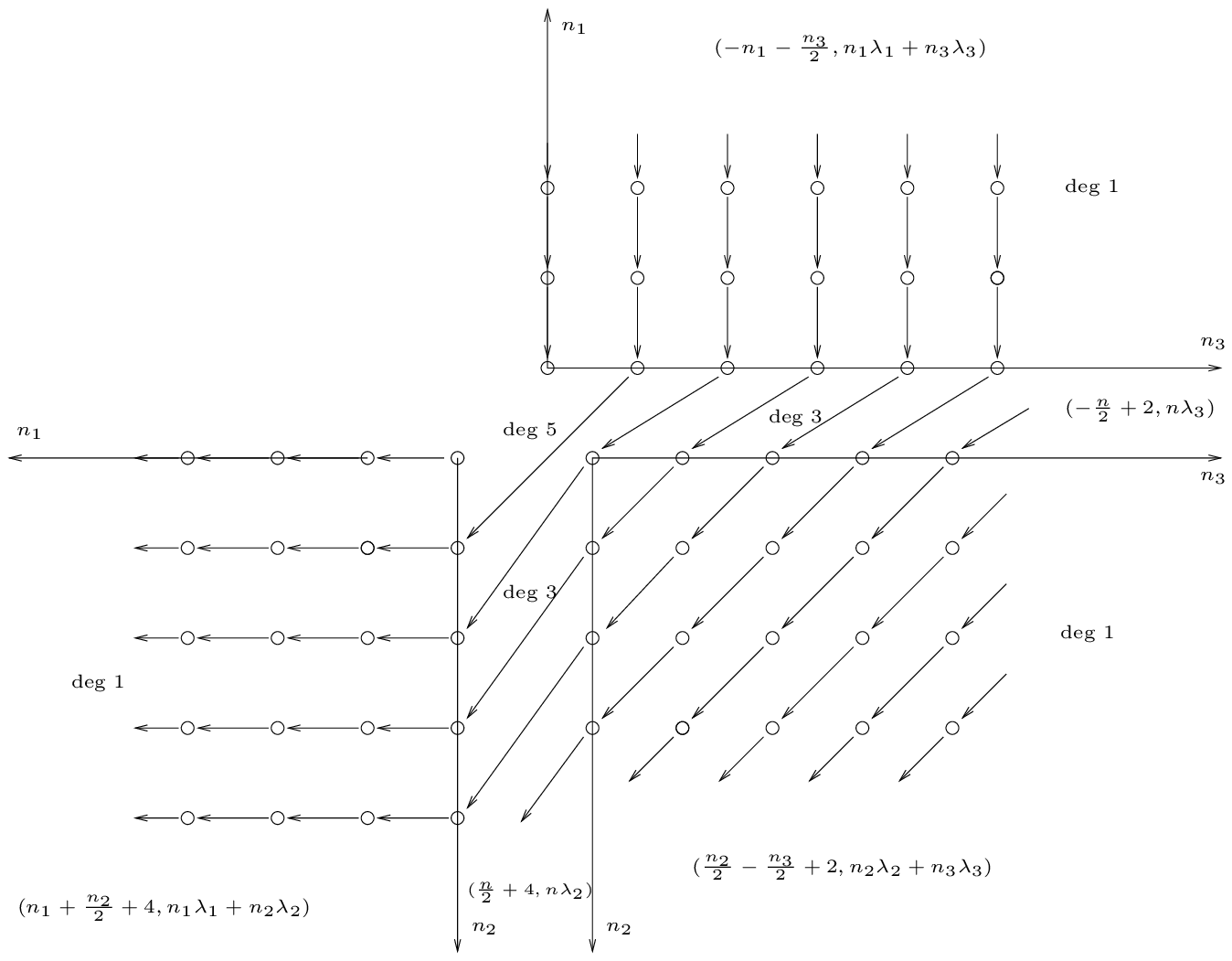}
\end{center}
%\caption{Fig: 2}
\end{figure}

% \vskip .5cm

% %\

% %\

%%%%%%%%%%%%%%%%%%%%%%%%%%%%%%%%%%%%%%%%%%%%%%%%%%%%

\section{Formal distributions, Lie conformal superalgebras and their
modules}\label{sec:formal}

%%%%%%%%%%%%%%%%%%%%%%%%%%%%%%%%%%%%%%%%%%%%%%%%%%%%

In this section we introduce the basic definitions and notations
in order to have a more or less self-contained work, for details
see \cite{BKL}, \cite{FK}, \cite{K1} and references there in.

\begin{definition} A {\it   Lie conformal superalgebra} $R$ is  a left
$\ZZ/2\ZZ$-graded $\cp$-module endowed with a $\CC$-linear map
$R\otimes R  \longrightarrow \CC[\la]\otimes R$, $ a\otimes b
\mapsto [a_\la b]$, called the $\la$-bracket, and  satisfying the
following axioms $(a,\, b,\, c\in
  R)$:

\

\noindent Conformal sesquilinearity $ \qquad  [\pa a_\la b]=-\la
[a_\la b],\qquad  [a_\la \pa
  b]=(\la+\pa) [a_\la b]$,

\vskip .3cm

\noindent Skew-symmetry $\ \qquad\qquad\qquad [a_\la
b]=-(-1)^{p(a)p(b)}[b_{-\la-\pa} \ a]$,

\vskip .3cm

\noindent Jacobi identity $\quad\qquad\qquad\qquad [a_\la [b_\mu
c]]=[[a_\la b]_{\la+\mu} c] + (-1)^{p(a)p(b)}[b_\mu [a_\la c]]. $

\vskip .5cm

Here and further $p(a)\in \ZZ/2\ZZ$ is the parity of $a$.
\end{definition}

A Lie conformal superalgebra is called $finite$ if it has finite
rank as a
  $\CC[\pa]$-module. The notions of homomorphism, ideal and subalgebras of
a Lie conformal superalgebra are defined in the usual way. A Lie
conformal superalgebra $R$ is $simple$ if $[R_\la R]\neq 0$ and
contains no ideals except for zero and itself.

\vfill
\pagebreak

\begin{definition} A {\it module} M over a Lie conformal superalgebra  $R$
is  a $\ZZ/2\ZZ$-graded $\cp$-module endowed with a
  $\CC$-linear map  $R\otimes M \longrightarrow \CC[\la]\otimes M$,
  $a\otimes v  \mapsto a_\la v$,
satisfying the following axioms $(a,\, b \in R),\ v\in M$,

\vskip .2cm

\noindent$(M1)_\la \qquad   (\pa a)_\la^M v= [\pa^M,
a_\la^M]v=-\la a_\la^Mv   ,$

\vskip .2cm

\noindent $(M2)_\la\qquad [a_\la^M, b_\mu^M]v=[a_{ \la}
b]_{\la+\mu}^Mv$.

\vskip .2cm

An $R$-module $M$ is called {\it finite} if it is finitely
generated over $\cp$. An $R$-module $M$ is called {\it
irreducible} if it contains no non-trivial submodule, where the
notion of submodule is the usual one.
\end{definition}

\

Given a   Lie conformal superalgebra $R$,  let $\tilde R= R[t, t^{-1}]$
with $\tilde\p=\p +\p_t$ and define the bracket \cite{K1}:
\begin{equation}
    \label{eq:1}
[a t^n ,  b t^m]=\sum_{j\in\ZZ_+}  \left(\begin{array}{c}
m\\j\end{array}\right)
  [a_{(j)}b] t^{m+n-j}.
\end{equation}
Observe that $\tilde\p\tilde R$ is an ideal of $\tilde R$ with
respect to  this bracket, and consider  the Lie superalgebra
$\Alg R=\tilde R /
\tilde\p\tilde R$ with this bracket.

\

An important tool for the study of Lie conformal superalgebras and
their modules is the (extended) annihilation superalgebra. The {\it
annihilation superalgebra} of a Lie conformal superalgebra $R$  is the
subalgebra $\A( R)$ (also denoted by $\Alg R_+$) of the Lie
superalgebra $\Alg R$ spanned by all elements $at^n$, where $a\in
R ,  n\in \ZZ_+$. It is clear from (\ref{eq:1}) that this is a
subalgebra, which is invariant with respect to the derivation
$\p=-\p_t$ of $\Alg R$. The {\it extended annihilation superalgebra} is
defined as
\begin{displaymath}
\A(R)^e=(\Alg R)^+:=\CC\p \ltimes (\Alg R)_+.
\end{displaymath}
Introducing the generating series
\begin{equation} \label{eq:la}
a_\la =\sum_{j\in\ZZ_+} \frac{\la^j}{j!} (a t^j), \,\,a \in R,
\end{equation}
we obtain from (\ref{eq:1}):
\begin{equation} \label{eq:2222}
[a_\la , b_\mu ] = [a_\la b]_{\la +\mu}, \quad \p(a_\la)= (\p a)_\la
=-\la a_\la.
\end{equation}

Formula (\ref{eq:2222}) implies the following important
proposition relating modules over a Lie conformal superalgebra $R$
to certain modules over the corresponding extended annihilation
superalgebra $(\Alg R)^+$.

\begin{proposition}
     \label{prop:1}
\cite{CK} A module over a Lie conformal superalgebra $R$ is the
same  as a
 module over the Lie superalgebra $(\Alg R)^+$  satisfying the property
\begin{equation} \label{eq:5}
a_\la m\in \CC[\la]\otimes M  \hbox{ \  for any } a\in R, m\in M.
\end{equation}
(One just views the action of the generating series $a_\la$ of
$(\Alg R)^+$ as the $\la$-action of $a\in R$).
\end{proposition}

The problem of classifying modules over a Lie conformal
superalgebra $R$ is thus reduced to the problem of classifying a
class of modules over the Lie superalgebra $(\Alg R)^+$.

Let $\fg$ be a Lie superalgebra satisfying the following three
conditions (cf. \cite{CL}, p.911):

(L1) $\fg$ is  $\ZZ$-graded of finite depth $d\in\NN$, i.e.
$g=\oplus_{j\geq -d} \fg_j$ and $[\fg_i , \fg_j ]\subset
\fg_{i+j}$.

(L2) There exists a semisimple element $z\in\fg_0$ such that its
centralizer in $\fg$ is contained in $\fg_0$.

(L3) There exists an element $\p\in\fg_{-d}$ such that $[\p,
g_i]=g_{i-d}$, for $i\geq 0$.

\vskip .2cm

Some examples of Lie superalgebras satisfying (L1)-(L3) are
provided by annihilation superalgebras of Lie conformal superalgebras.

If $\fg$ is the annihilation superalgebra of a Lie conformal
superalgebra, then the modules V over
%the annihilation superalgebra
$\fg$ that correspond to finite modules over the corresponding
Lie conformal superalgebra satisfy the following conditions:

(1) For all $v\in V$ there exists an integer $j_0\geq -d$ such
that $\fg_j v=0$, for all $j\geq j_0$.

(2) $V$ is finitely generated over $\CC[\p]$.

\noindent Motivated by this, the $\fg$-modules satisfying these
two properties are called {\it finite conformal modules}.

We have a triangular decomposition
\begin{equation} \label{eq:ddd}
\fg=\fg_{<0} \oplus\fg_0 \oplus\fg_{>0}, \qquad \hbox{ with }
\,\fg_{<0}=\oplus_{ j <0} \fg_j, \fg_{>0}=\oplus_{ j >0} \fg_j.
\end{equation}
Let $\fg_{\geq 0}=\oplus_{j\geq 0} \fg_j$. Given a $\fg_{\geq
0}$-module $F$, we may consider the associated induced
$\fg$-module
\begin{displaymath}
\Ind(F)=\Ind^\fg_{\fg_{\geq 0}} F=U(\fg)\otimes_{U(\fg_{\geq
0})} F,
\end{displaymath}
called the {\it generalized Verma module} associated to $F$.
We shall
%omit $\fg$ and $\fg_{\geq 0}$, and simply denote it as
identify Ind$(F)$ with $U(\fg_{<0})\otimes F$ via the PBW theorem.

Let $V$ be a $\fg$-module. The elements of the subspace
\begin{displaymath}
\Sing(V):=\{ v\in V| \fg_{>0} v=0\}
\end{displaymath}
are called {\it singular vectors}. For us the most important case
is when $V=\Ind(F)$. The $\fg_{\geq 0}$-module $F$ is canonically a
$\fg_{\geq 0}$-submodule of $\Ind(F)$, and $\Sing(F)$ is a subspace of
$\Sing(\Ind(F))$, called the {\it subspace of trivial singular
vectors}. Observe that $\Ind(F)= F\oplus F_+$, where $F_+=U_+(\fg_{<
0})\otimes F$ and $U_+(\fg_{< 0})$ is the augmentation ideal of
the algebra $U(\fg_{< 0})$. Then non-zero elements of the space
\begin{displaymath}
 \Sing_+(\Ind(F)):=\Sing(\Ind(F))\cap F_+
\end{displaymath}
are called {\it non-trivial singular vectors}. The following simple key
result will be used in the rest of the paper, see \cite{KR1, CL}.

\vfill
\pagebreak

\begin{theorem}
     \label{th:1}
 Let $\fg$ be a Lie superalgebra   that
satisfies (L1)-(L3).

(a) If $F$ is an irreducible finite-dimensional $\fg_{\geq
0}$-module, then the subalgebra $\fg_{> 0}$ acts trivially on $F$
and $\Ind(F)$ has a unique maximal submodule.
%Therefore, the
%irreducible finite-dimensional $\fg_{\geq 0}$-modules come from
%irreducible finite-dimensional $\fg_0$-modules that are extended
%trivially on $\fg_{> 0}$.

(b) Denote by $\Ir(F)$ the quotient  by the unique maximal submodule
of $\Ind(F)$. Then the map $F\mapsto \Ir(F)$ defines a bijective
correspondence between irreducible finite-dimensional
$\fg_0$-modules and irreducible finite conformal  $\fg$-modules.

(c) A $\fg$-module $\Ind(F)$ is irreducible if and only if the
$\fg_0$-module $F$ is  irreducible and $\Ind(F)$ has no non-trivial
singular vectors.
\end{theorem}

\

In the following section we will describe the Lie conformal
superalgebra $CK_6$ and its annihilation superalgebra $E(1,6)$. In
the remaining sections we shall study the induced $E(1,6)$-modules
and its singular vectors in order to apply Theorem \ref{th:1} to
get the classification of irreducible finite modules over the Lie
conformal algebra $CK_6$.

\

\

%%%%%%%%%%%%%%%%%%%%%%%%%%%%%%%%%%%%%%%%
\section{Lie conformal superalgebra $CK_6$, annihilation Lie superalgebra
$E(1,6)$ and the induced modules} \lbb{sK}
%%%%%%%%%%%%%%%%%%%%%%%%%%%%%%%%%%%%%%%%

\

Let $\Lambda(n)$ be the Grassmann superalgebra in the $n$ odd
indeterminates $\xi_1, \xi_2,\ldots , \xi_n$. Let $t$ be an even
indeterminate, and let $\Lambda(1,n)_+=\CC[t]\otimes \Lambda(n)$.
The  Lie conformal superalgebra $K_n$ can be identified with
\begin{equation} \label{eq:3.2}
K_n=\CC[\p]\otimes \Lambda(n),
\end{equation}
the $\la$-bracket for
%$ f,g\in \Lambda(n)$
$f=\xi_{i_1} \dots \xi_{i_r}, g=\xi_{j_1} \dots \xi_{j_s}$ being
as follows \cite{FK}:
\begin{equation} \label{eq:3.3}
[f_\la g]=\bigg( (r-2)\p (fg) + (-1)^r \sum_{i=1}^n (\p_i f)(\p_i
g)\bigg) + \lambda (r+s-4)fg.
\end{equation}

The annihilation Lie superalgebra of $K_n$ can be identified with
(see \cite{BKL})
\begin{equation} \label{eq:3.4}
\A(K_n)=K(1,n)_+ =\Lambda(1,n)_+,
\end{equation}
with the corresponding Lie bracket for elements
$f,g\in\Lambda(1,n)$ being
\begin{displaymath}
[f,g]=\bigg( 2f - \sum_{i=1}^n \xi_i \p_i f\bigg) (\p_t g) - (\p_t
f)\bigg( 2g - \sum_{i=1}^n \xi_i \p_i g\bigg) + (-1)^{p(f)}
\sum_{i=1}^n (\p_i f)(\p_i g).
\end{displaymath}
The extended annihilation superalgebra is
\begin{displaymath}
\A(K_n)^e= K(1,n)^+=\CC \, \p \ltimes K(1,n)_+,
\end{displaymath}
where $\p$ acts on it as $-\hbox{ad}\,\p_t$. Note that $\A(K_n)^e$
is isomorphic to the direct sum of $\A(K_n)$ and the trivial
1-dimensional Lie algebra $\CC (\p+\dfrac{1}{2})$.

We define in $K(1,n)_+$ a gradation by putting
\begin{displaymath}
\deg(t^m \xi_{i_1} \dots \xi_{i_k})=2m+k-2,
\end{displaymath}
making it a $\ZZ$-graded Lie superalgebra of depth 2:
$K(1,n)_+=\oplus_{j\geq -2} (K(1,n)_+)_j$.  It is easy to check
that $K(1,n)_+$ satisfies conditions (L1)-(L3).

We introduce the following notation:
\begin{align} \label{eq:I}
\xi_I & :=\xi_{i_1}\dots \xi_{i_k}, \qquad\quad \hbox{  if }\quad
 I=\{i_1,\dots , i_k\},\nonumber\\
|f| & :=k \qquad \qquad \quad \qquad \,\hbox{if }\quad
f=\xi_{i_1}\dots \xi_{i_k} .\nonumber
\end{align}

For  a monomial $\xi_I\in \La (n)$, we let ${\xi_I}^*$ be its
Hodge dual, i.e. the unique monomial in $\La (n)$ such that
${\xi_I} \xi_I^*= \xi_1 \dots \xi_n$.

{\bf Warning:} this definition corresponds to the one in \cite
{CK} or \cite{CL}, pp. 922, but in \cite{BKL} Theorem 4.3, the
Hodge dual was defined in a different way.

The Lie conformal superalgebra $CK_6$ is defined as the subalgebra
of $K_6$ given by (cf. \cite{CK2}, Theorem 3.1)
\begin{equation*}\label{CK6}
    CK_6=\cp \hbox{-span}\  \{f-i(-1)^{\frac{|f|(|f|+1)}{2}}(-\p)^{3-|f|}
    \, f^*\, : \, f\in \La (6), 0\leq |f|\leq 3\}.
\end{equation*}
Now, we define a linear operator $A:K(1,6)_+\rightarrow K(1,6)_+$
by (cf. \cite{CK3}, p.267)
\begin{equation}\label{A}
    A(f)=(-1)^{\frac{d(d+1)}{2}}\left(\frac{d}{d
    t}\right)^{3-d}(f^*),
\end{equation}
where $f$ is a monomial in $K(1,6)_+$, $d$ is the number of odd
indeterminates in $f$, the operator $(\frac{d}{dt})^{-1}$
indicates integration with respect to $t$ (i.e. it sends $t^n$
to $t^{n+1}/(n+1)$), and $f^*$ is the Hodge
dual of $f$. Then, the annihilation Lie superalgebra $E(1,6)$ of
$CK_6$ is identified with the subalgebra of $K(1,6)_+$ given by
the image of the operator $I-iA$. Since the linear map $A$
preserve the $\ZZ$-gradation, the subalgebra $E(1,6)$ inherits the
$\ZZ$-gradation.

\

Using Theorem \ref{th:1}, the classification of finite irreducible
$CK_6$-modules can be reduced to the study of induced modules for
$E(1,6)$. Observe that the graded subspaces of $E(1,6)$ and
$K(1,6)_+$ with non-positive degree are the same. Namely,
\begin{align} \label{eq:n1}
  E(1,6)_{-2} & = < \{ {\bf 1} \} >, \nonumber\\
  E(1,6)_{-1} & =< \{ \xi_i  \ : \ 1\leq i\leq 6\} > \\
 E(1,6)_{0} & =< \{ t\}\cup \{ \xi_i\xi_j  \ : \ 1\leq i<j\leq 6\} >
  \nonumber
\end{align}
We shall use the following notation for the basis elements of
$E(1,6)_{0}$ (cf. \cite{BKL}):
\begin{equation} \label{eq:so}
E_{00}=t, \qquad\quad F_{ij}= -\xi_i\xi_j.
\end{equation}
Observe that $E(1,6)_{0}\simeq \CC E_{00}\oplus \so(6)\simeq
 \cso(6)$. Take
\begin{equation} \label{eq:del}
\p:=-\frac{1}{2}{\bf 1}
\end{equation}
as the element that satisfies (L3) in  section \ref{sec:formal}.

For the rest of this work, $\fg$ will be $E(1,6)$. Let $F$ be  a
finite-dimensional irreducible  $\fg_0$-module, which we extend to
a $\fg_{\geq 0}$-module by letting $\fg_j$ with $j>0$ acting
trivially. Then we shall identify, as above:
\begin{equation} \label{eq:indu}
\Ind (F)\simeq  \Lambda(1,6)_+\otimes F\simeq
\cp\otimes\La(6)\otimes F
\end{equation}
as $\CC$-vector spaces.

Since the non-positive graded subspaces of $E(1,6)$ are the same
as those of $K(1,6)_+$, the $\la$-action is given by restricting
the $\la$-action in  Theorem 4.1 in \cite{BKL}. In the following
theorem, we describe the $\fg$-action of $K(1,6)_+$ on $\Ind (F)$
using the $\lambda$-action notation in (\ref{eq:la}), i.e.
\begin{equation}\label{lam}
f_\la(g\otimes v)= \sum_{j\geq 0} \frac{\la^j}{j!} \ (t^j f) \cdot
(g\otimes v)
\end{equation}
for $f, g\in\La(6)$ and $v\in F$.

\

\begin{theorem} \label{th:action} For any monomials
$f, g\in \La(6)$ and $v \in F$, where $F$ is a $\cso(6)$-module,
we have the following formula for the $\lambda$-action of
$K(1,6)_+$ on $\Ind(F)$:
\begin{align*}
 & f_\la(g\otimes v) =  \\
 & = (-1)^{p(f)} (|f|-2) \p (\p_f g)\otimes v
 + \sum_{i=1}^6 \p_{(\p_i f)} (\xi_i g)\otimes v
 + (-1)^{p(f)} \sum_{ r<s} \p_{(\p_r\p_s f)}g \otimes F_{rs}v\\
 & +\la \bigg[ (-1)^{p(f)}(\p_f g)\otimes E_{00} v
 + (-1)^{p(f)+p(g)} \sum_{i=1}^6 \big(\p_f(\p_i g)\big)\xi_i\otimes v
 + \sum_{i\neq j} \p_{(\p_i f)}(\p_j g)\otimes F_{ij} v  \bigg]\\
  & + \la^2 (-1)^{p(f)} \sum_{i<j} \p_f (\p_i\p_j g) \otimes F_{ij} v.
\end{align*}
\end{theorem}

In the last part of this section we shall state an easier formula
for the $\la$-action in the induced module (see Theorem 4.3,
\cite{BKL}). This is done by taking {\bf the Hodge dual of the
basis modified by a sign}, since we are using the definition of
Hodge dual given in \cite{CK}, instead of the one used in
\cite{BKL}. Namely, let $T$ be the vector space automorphism of
Ind$(F)$ given by $T(g\otimes v)=(-1)^{|g|}{g}^*\otimes v$, then
the following theorem  gives the formula for the composition
$T\circ (f_\la\,\cdot)\circ T^{-1}$.

\begin{theorem} \label{th:action-dual} Let $F$ be a
$\cso(6)$-module. Then the $\lambda$-action of $K(1,6)_+$ in
$\Ind(F)=\cp\otimes \Lambda(6)\otimes F$, given by Theorem
\ref{th:action}, is equivalent to the following one:
\begin{align*} \label{eq:ac}
 & f_\la(g\otimes v) =  (-1)^{\frac{|f|(|f|+1)}{2}+ |f||g|} \ \times \\
 & \times\   \Bigg\{(|f|-2) \p (f g)\otimes v
  - (-1)^{p(f)} \sum_{i=1}^6 (\p_i f) (\p_i g)\otimes v
  -  \sum_{r<s} (\p_r\p_s f)g \otimes F_{rs}v\\
 & +\la \bigg[ f g\otimes E_{00} v
 - (-1)^{p(f)} \sum_{i=1}^6 \p_i \big(f \xi_i g\big) \otimes v
 + (-1)^{p(f)}\sum_{i\neq j} (\p_i f)\xi_j g \otimes F_{ij} v  \bigg]\\
  & - \la^2  \sum_{i<j} f \xi_i\xi_j g \otimes F_{ij} v \Bigg\}.
\end{align*}
\end{theorem}

\

%%%%%%%%%%%%%%%%%%%%%%%%%%%%%%%%%%%%%%%%
\section{Singular vectors}\lbb{ssingular}
%%%%%%%%%%%%%%%%%%%%%%%%%%%%%%%%%%%%%%%%

\

By Theorem \ref{th:1}, the classification of irreducible finite
modules over the Lie conformal superalgebra $CK_6$ reduces to the
study of singular vectors in the induced modules $\Ind(F)$, where
$F$ is an irreducible finite-dimensional $\cso(6)$-module. This
section will be devoted to the classification  of singular
vectors.

When we discuss the highest weight of vectors and singular
vectors, we always mean with respect to the upper Borel subalgebra
in $E(1,6)$ generated by $(E(1,6))_{>0}$ and the elements of the
Borel subalgebra of $\so(6)$ in $E(1,6)_0$. More precisely, recall
(\ref{eq:so}), where we defined $F_{ij}=-\xi_i\xi_j\in E(1,6)_0
\simeq \CC E_{00}\oplus\mathfrak{so}(6)$. Observe that $F_{ij}$
corresponds to $E_{ij}-E_{ji}\in \mathfrak{so}(6)$, where $E_{ij}$
are the elements of the standard basis of matrices. Consider the
following (standard) notation for $\mathfrak{so}(6,\CC)$ (cf.
\cite{Knapp}, p.83): We take
\begin{equation}\label{H}
H_j=i \ F_{2j-1,2j}, \qquad 1\leq j\leq 3,
\end{equation}
a basis of  a Cartan subalgebra $\mathfrak{h}_0$. Let $\ep_j\in
\mathfrak{h}_0^*$ be given by $\ep_j(H_k)=\delta_{jk}$. Let
$$
\Delta=\{\pm\ep_i\pm\ep_j\ |\ i< j\}
$$
be the set of roots. The root space decomposition is
$$
\mathfrak{g}=\mathfrak{h}_0\oplus\bigoplus_{\alpha\in \Delta}
\mathfrak{g}_\alpha, \qquad \hbox{ with } \fg_\al=\CC E_\al
$$
where, for $1\leq l< j\leq 3$,
\begin{align}\label{E}
E_{\ep_l-\ep_j} \, \  & =F_{2l-1,2j-1}+F_{2l,2j}+i
(F_{2l-1,2j}-F_{2l,2j-1}),\nonumber\\
E_{\ep_l+\ep_j} \, \  & =F_{2l-1,2j-1}-F_{2l,2j}- i (F_{2l-1,2j}+
F_{2l,2j-1}),\\
E_{-(\ep_l-\ep_j)} & =F_{2l-1,2j-1}+F_{2l,2j}-i
(F_{2l-1,2j}-F_{2l,2j-1}),\nonumber\\
E_{-(\ep_l+\ep_j)} & =F_{2l-1,2j-1}-F_{2l,2j}+ i (F_{2l-1,2j}+
F_{2l,2j-1}), \nonumber
\end{align}
\

Let $\Pi=\{\ep_1-\ep_2,  \ep_{2}-\ep_3, \ep_2+\ep_3\}$ and $
\Delta^+=\{\ep_i\pm\ep_j\ |\ i < j\}$, be the simple and positive
roots respectively. Consider
\begin{equation*}
\al_{lj}:=F_{2l-1,2j-1}-i
F_{2l,2j-1}=\frac{1}{2}(E_{\ep_l-\ep_j}+E_{\ep_l+\ep_j})
\end{equation*}
\begin{equation}\label{eq:beta}
\be_{lj}:=F_{2l,2j}+i
F_{2l-1,2j}=\frac{1}{2}(E_{\ep_l-\ep_j}-E_{\ep_l+\ep_j})\qquad
\end{equation}

\vskip .2cm

\noindent Then, the Borel subalgebra is
\begin{equation}\label{eq:borel}
    B_{\so (6)}=<\{\al_{lj},\be_{lj} \ |\ 1\leq i< j\leq
    3 \}>.
\end{equation}

\

Recall that the Cartan subalgebra  $\fh$ in $(CK(1,6)_+)_0\simeq
\CC E_{00}\oplus\mathfrak{so}(6)\simeq \mathfrak{cso}(6)$ is
spanned by the elements $E_{00} , H_1 , H_2 , H_3$. Let  $e_0\in
\fh^*$ be the linear functional given by $e_0(E_{00})=1$ and
$e_0(H_i)=0$ for all $i$. Let

\begin{align}
\la_1 & =\ep_1 \nonumber\\
\la_2 & =\frac{\ep_1+\ep_2-\ep_3}{2} \\
\la_3 & =\frac{\ep_1+\ep_2+\ep_3}{2} \nonumber
\end{align}
be the fundamental weights of $\so(6)$, extended to $\fh^*$ by
letting $\la_i(E_{00})=0$. That is $<\la_i,\alpha_j>=\delta_{ij}$,
where
$\alpha_1=\ep_1-\ep_2,\alpha_2=\ep_2-\ep_3,\alpha_3=\ep_2+\ep_3$.
Then the highest weight $\mu$ of the finite irreducible
$\cso(6)$-module $F_\mu$ can be written as
\begin{equation}\label{mmmm}
    \mu= n_0 e_0+ n_1\la_1+n_2\la_2+n_3\la_3,
\end{equation}
where $n_1, n_2$ and $n_3$ are non-negative integers. In order to
write explicitly weights for vectors in $CK(1,6)_+$-modules, we
will consider the notation (\ref{mmmm}).

\vskip .5cm

Consider  a singular vector $\vec{m}$ in the  $CK(1,6)_+$-module
$\Ind(F)=\CC[\p]\otimes \Lambda(6)\otimes F$, where $F$ is an
irreducible $\mathfrak{cso}(6)$-module, and the $\lambda$-action
of $CK(1,6)_+$ in $\Ind(F)$ is given by restricting the
$\la$-action of $K(1,6)_+$ in Theorem \ref{th:action-dual}.

\vskip .5cm

Using (\ref{A}), (\ref{lam}), and the description of the upper
Borel subalgebra of $E(1,6)$, we obtain  that a vector $\vec{m}$
in the $E(1,6)$-module $\Ind(F)$
 is a singular  vector if and only if the following
conditions are satisfied

\

(S1) For all $f\in\Lambda(6)$, with $0\leq |f| \leq 3$,
\begin{equation*}\label{}
\frac{d\,^2}{d\la^2}\,\left(f\,_\la
\,\vec{m}-i(-1)^{\frac{|f|(|f|+1)}{2}}\la^{3-|f|}(f^*\,_\la
\,\vec{m})\right)=0.
\end{equation*}

\vskip .3cm

(S2) For all $f\in\Lambda(6)$, with $1\leq |f| \leq 3$,
\begin{equation*}\label{ }
\frac{d}{d\la}\,\left(f\,_\la
\,\vec{m}-i(-1)^{\frac{|f|(|f|+1)}{2}}\la^{3-|f|}(f^*\,_\la
\,\vec{m})\right)|_{\la=0}=0.
\end{equation*}

\vskip .3cm

(S3) For all $f$ with $|f|= 3$ or $f\in B_{\so(6)}$,
\begin{equation*}\label{ }
\left(f\,_\la
\,\vec{m}-i(-1)^{\frac{|f|(|f|+1)}{2}}\la^{3-|f|}(f^*\,_\la
\,\vec{m})\right)|_{\la=0}=0.
\end{equation*}

\

In order to classify the finite irreducible $CK_6$-modules we
should solve the equations (S1-S3) to obtain the singular vectors.

\

Observe that the $\mathbb{Z}$-gradation in $E(1,6)$, translates into a
$\mathbb{Z}_{\leq 0}$-gradation in Ind$(F)=\cp\otimes \Lambda(6)\otimes F$ where wt $\partial=-2$ and wt $\xi_i=-1$, but we shall work with the $\la$-action given by Theorem \ref{th:action-dual} where we considered the (modified) Hodge dual basis, therefore from now on the $\mathbb{Z}_{\leq 0}$-gradation in Ind$(F)$ is given by wt $\partial=-2$ and wt $\xi_I=6-|I|$.

\

The next theorem is one of the main result of this section and
gives us the complete classification of singular vectors:

\vskip .5cm

\begin{theorem}\label{sing-vect} Let $F_\mu$ be an
irreducible finite-dimensional $\cso(6)$-module with highest
weight $\mu$. Then, there exists  a non-trivial singular vector
$\vec m$  in the $E(1,6)$-module $\Ind(F_\mu)$ if and only if the
highest weight $\mu$  is one of the following:

\vskip .3cm

\begin{itemize}
    \item[(a)]  $\mu=\left(\dfrac{1}{2}+4\right)e_0+\la_2,$, where
 (up to scalar) $\vec m=\p^2 \ g_5 \ +\ \p \ g_3 +g_1$ of  degree
-5, given by (\ref{m5-m}), with singular weight

   $$
     \left(-\dfrac{1}{2}\right)e_0+\lambda_3,
     $$
      \vskip .3cm

\item[(b)] $
     \mu=\left(\dfrac{n}{2}+4\right)e_0+n\la_2,$ with  $n\geq 2$, where
 (up to scalar)
     $\vec m= \sum_{|I|= 3} \xi_I\otimes v_{I}$ of  degree -3,
     given by (\ref{m3-m}), with singular weight

     $$
     \left(\dfrac{n}{2}+1\right)e_0+(n-2)\la_2,
  $$
      \vskip .3cm

\item[(c)] $
     \mu=\left(-\dfrac{n}{2}+2\right)e_0+n\la_3,$ with  $n\geq 0$,
     where
 (up to scalar)
     $\vec m= \sum_{|I|= 3} \xi_I\otimes v_{I}$ of  degree -3,
     given by (\ref{m3-sing-2}), with singular weight

     $$
     \left(-\dfrac{n}{2}-1\right)e_0+(n+2)\la_3,
  $$
      \vskip .3cm

    \item[(d)] $
     \mu=\left(n_1+\dfrac{n_2}{2}+4\right)e_0+n_1\la_1+n_2 \la_2,$
     with  $n_1\geq 1$, $n_2\geq 0$, where
 (up to scalar)
 $\vec m= \sum_{|I|= 5} \xi_I\otimes v_{I}$ of  degree -1,
     given by (\ref{m1-sing-1}), with singular weight

     $$
     \mu=\left(n_1+\dfrac{n_2}{2}+3\right)e_0+(n_1-1)\la_1+n_2 \la_2,
$$

 \vskip .3cm

    \item[(e)]  $
     \mu=\left(\dfrac{n_2}{2}-\dfrac{n_3}{2}+2\right)e_0+n_2\la_2+n_3 \la_3,$
     with  $n_2\geq 1$, $n_3\geq 0$,
 where
 (up to scalar) $\vec m= \sum_{|I|= 5} \xi_I\otimes v_{I}$ of  degree -1,
     given by (\ref{m1-sing-2}), with singular weight

     $$
     \mu=\left(\dfrac{n_2}{2}-\dfrac{n_3}{2}+1\right)e_0+(n_2-1)\la_2+(n_3+1) \la_3,
$$

 \vskip .3cm

\item[(f)] $
     \mu=-\left(n_1+\dfrac{n_3}{2}\right)e_0+n_1\la_1+n_3 \la_3,$ with  $n_1\geq 0$, $n_3\geq 0$.
 where
 (up to scalar) $\vec{m}=\big(\xi_{\{2\}^c}-i \xi_{\{1\}^c}\big)\otimes
v_{\mu}$ of  degree -1, with singular weight

     $$
     \mu=-\left(n_1+\dfrac{n_3}{2}+1\right)e_0+(n_1+1)\la_1+n_3 \la_3,
$$
      \vskip .3cm
\end{itemize}
\end{theorem}

\

\begin{remark} (a) We have the explicit expression of all singular vectors
 in terms of the highest weight vector $v_{\mu}$ in each case.
\vskip.3cm

\noindent (b) The  family (f) of singular vectors with $n_3=0$
corresponds to the first family of singular vectors in $K_6$,
computed in Theorem 5.1 \cite{BKL}.

\vskip.3cm

\noindent (c) The  family (d) of singular vectors with $n_2=0$
corresponds to the second family of singular vectors in $K_6$,
computed in Theorem 5.1 \cite{BKL}.

\vskip.3cm

\noindent (d) The  family (e) is new, with no analog in $K_6$.

\vskip.3cm

\noindent (e) The highest weight where we have a singular vector
of degree -5 correspond to the case $k=\dfrac{1}{2}$ in  the
family (b), and the one parameter families (b) and (c) are the
(not present) cases $k=l$ in the families (d) and (e)
respectively.

\end{remark}

\

\noindent Using  Theorem \ref{sing-vect}, we obtain the following
theorem that is the main result of this work and gives us the
complete list of highest weights of degenerate irreducible
modules:

\

\begin{theorem}\label{degenerate}
Let $F_\mu$ be an irreducible finite-dimensional $\cso(6)$-module
with highest weight $\mu$. Then the $E(1,6)$-module $\Ind(F_\mu)$
is degenerate if and only if $\mu$  is one of the following:

\vskip .3cm

\begin{itemize}

    \item[(a)]  $
     \mu=\left(n_1+\dfrac{n_2}{2}+4\right)e_0+n_1\la_1+n_2 \la_2,$ with  $n_1\geq 1$, $n_2\geq 0$ or $n_1=0, n_2\geq 1$,

 \vskip .3cm

    \item[(b)] $
     \mu=\left(\dfrac{n_2}{2}-\dfrac{n_3}{2}+2\right)e_0+n_2\la_2+n_3 \la_3,$ with  $n_2\geq 0$, $n_3\geq 0$,

 \vskip .3cm

\item[(c)] $
     \mu=-\left(n_1+\dfrac{n_3}{2}\right)e_0+n_1\la_1+n_3 \la_3,$ with  $n_1\geq 0$, $n_3\geq 0$.

      \vskip .3cm
\end{itemize}
\end{theorem}

\

\vskip 1cm

The rest of this section together with the Appendix A are devoted
to the proof of this theorem. The proof will be done through
several lemmas.

\

Recall that the Cartan subalgebra  $\fh$ in $(CK(1,6)_+)_0\simeq
\CC E_{00}\oplus\mathfrak{so}(6)\simeq \mathfrak{cso}(6)$ is
spanned by the elements
$$
E_{00} , H_1 , H_2 , H_3,
$$
and, for technical reasons as in our work \cite{BKL}, from now on
we shall write the weights of an eigenvector for the Cartan
subalgebra $\fh$ as an $1+3$-tuple for the corresponding
eigenvalues of this basis:
\begin{equation}\label{lambda-mu}
    \mu=(\mu_0; \mu_1,\mu_2,\mu_3).
\end{equation}

Let $\mu= n_0 e_0+ n_1\la_1+n_2\la_2+n_3\la_3$ be  the highest
weight of the finite irreducible $\cso(6)$-module $F_\mu$, where
$n_1, n_2$ and $n_3$ are non-negative integers, as in
(\ref{mmmm}). Using the notation (\ref{lambda-mu}), this highest
weight  can be written as the
 $1+3$-tuple
\begin{equation}\label{mimi}
\mu=\left(n_0;  n_1+\frac{n_2}{2}+\frac{n_3}{2}\  , \
\frac{n_2}{2}+\frac{n_3}{2}\  ,\  -\frac{n_2}{2}+\frac{n_3}{2}
\right).
\end{equation}

\

Let $\vec{m}\in \hbox{Ind}(F)=\CC[\p]\otimes \Lambda(6)\otimes F$
be a singular vector, then
\begin{equation*}
\vec{m}=\sum_{k=0}^N\sum_I \p^k (\xi_I\otimes v_{I,k}), \quad
\hbox{ with } v_{I,k}\in F.
\end{equation*}

\begin{lemma} \label{lem:deg} If $\vec{m} \in \hbox{\rm Ind}(F)$ is a
singular vector, then the degree of $\vec{m}$ in $\p$ is at most
2. Moreover, any singular vector have this form:
$$
\vec{m}= \p^2 \ \sum_{|I|\geq 5} \xi_I\otimes v_{I,2}\  +\  \p \
\sum_{|I|\geq 3} \xi_I\otimes v_{I,1} +\sum_{|I|\geq 1}
\xi_I\otimes v_{I,0}.
$$
\end{lemma}

\begin{proof} the proof of this lemma will be published in a
second paper of the series of three papers.
\end{proof}

\

The $\mathbb{Z}$-gradation in $E(1,6)$, translates into a
$\mathbb{Z}_{\leq 0}$-gradation in Ind$(F)$:

\begin{align*}
\mathrm{Ind}(F) & \simeq  \Lambda(1,6)\otimes F\simeq
\cp\otimes\La(6)\otimes F \\ & \simeq \underbrace{\mathbb{C}\
1\otimes F}\ \oplus\ \underbrace{\mathbb{C}^6\otimes F}\ \oplus\
\underbrace{((\mathbb{C}\ \p \otimes F)\oplus
(\Lambda^2(\mathbb{C}^6)\otimes F))}\ \oplus \cdots\\
& \qquad \ \mathrm{deg }\ 0\quad \ \ \quad \hbox{deg -1}\ \qquad
\qquad \quad \ \ \hbox{deg -2}\
\end{align*}

\noindent Therefore, in the previous lemma, we have proved that
any singular vector must have degree at most  -5.

Recall that in the theorem that gives us the $\la$-action, we
considered the Hodge dual of the natural bases in order to
simplify the formula of the action. Hence, any singular vector
must have one of the following forms:

\
\begin{align}\label{m-singular}
& \ \vec{m}= \p^2 \ \sum_{|I|= 5} \xi_I\otimes v_{I,2}\ +\ \p \
\sum_{|I|= 3} \xi_I\otimes v_{I,1} +\sum_{|I|= 1} \xi_I\otimes
v_{I,0}, \hbox{ (Degree -5).}\nonumber\\
& \ \vec{m}= \p^2 \ \sum_{|I|= 6} \xi_I\otimes v_{I,2}\ +\  \p \
\sum_{|I|= 4} \xi_I\otimes v_{I,1} +\sum_{|I|= 2} \xi_I\otimes
v_{I,0}, \hbox{ (Degree -4).}\nonumber\\
& \ \vec{m}=   \p \ \sum_{|I|= 5} \xi_I\otimes v_{I,1} +\sum_{|I|=
3} \xi_I\otimes v_{I,0}, \hbox{ (Degree -3).}\nonumber\\
& \  \vec{m}=   \p \ \sum_{|I|= 6} \xi_I\otimes v_{I,1}
+\sum_{|I|= 4} \xi_I\otimes v_{I,0}, \hbox{ (Degree -2).}\\
& \  \vec{m}=  \sum_{|I|= 5} \xi_I\otimes v_{I,0}, \hbox{ (Degree
-1).}\nonumber
\end{align}

\vskip .6cm

Now, we shall introduce a very important notation. Observe that
the formula for the action given by Theorem \ref{th:action-dual}
have the form
$$
f_\la (g\otimes v)= \p\  a + b + \la\ B + \la^2\ C= (\la +\p )\  a
+ b + \la\ (B-a) + \la^2\ C,
$$
by taking the coefficients in $\p$ and $\la^j$. Using it, we can
write the $\la$-action on the singular vector $\vec{m}=\p^2\
m_2+\p \ m_1+ m_0$ of degree 2 in $\p$, as follows
\begin{align}\label{eq:ffff}
 f_\la \vec{m} = &  \bigg[(\la +\p )\  a_0
+ b_0 + \la\ (B_0-a_0) + \la^2\ C_0\bigg] \nonumber\\
 &  + (\la +\p )\ \bigg[(\la +\p )\  a_1
+ b_1 + \la\ (B_1-a_1) + \la^2\ C_1\bigg]  \\
  &   + (\la +\p )^2 \ \bigg[(\la +\p )\  a_2
+ b_2 + \la\ (B_2-a_2) + \la^2\ C_2\bigg]. \nonumber
\end{align}
Obviously, these coefficients depend also in $f$ and $m$, and
sometimes we shall write for example $a_2(f)$ or $a(f,m_2)$,
instead of $a_2$, to emphasize the dependance, but we will keep it
implicit in the notation if no confusion may arise. In a similar
way, for the $\la$-action of $f^*$ on $\vec{m}=\p^2\ m_2+\p \ m_1+
m_0$ we use the notation
\begin{align}\label{eq:fd}
 f^*_\la \vec{m} = &  \bigg[(\la +\p )\  ad_0
+ bd_0 + \la\ (Bd_0-ad_0) + \la^2\ Cd_0\bigg] \nonumber\\
 &  + (\la +\p )\ \bigg[(\la +\p )\  ad_1
+ bd_1 + \la\ (Bd_1-ad_1) + \la^2\ Cd_1\bigg]  \\
  &   + (\la +\p )^2 \ \bigg[(\la +\p )\  ad_2
+ bd_2 + \la\ (Bd_2-ad_2) + \la^2\ Cd_2\bigg]. \nonumber
\end{align}
As before, these coefficients depend also in $f^*$ and $m$, and
sometimes we shall write for example $ad_2(f)$ or $ad(f,m_2)$,
instead of $ad_2$, to emphasize the dependance, but we will keep
it implicit in the notation if no confusion may arise.

\begin{lemma}\label{conditions} Let  $\vec{m}=\p^2\
m_2+\p \ m_1+ m_0$ be  a vector of degree at most -5. The
conditions (S1)-(S3) on $\vec{m}$ are equivalent to the following
list of equations

\vskip 0.3cm

\noindent $\underline{\bullet \hbox{ For  } |f|=0:}$

\begin{align}
 \mathit{C_0} & = - \mathit{B_1},\label{ec-a-1}\\
2\ \mathit{B_2} & = - {\displaystyle \mathit{a_2}}  -
{\displaystyle \mathit{C_1}} ,\\
2\ \mathit{bd_0} & = \,i\,\mathit{a_2} -  \,i \,\mathit{C_1}.
\end{align}

\vskip .1cm

\noindent $\underline{\bullet \hbox{ For  } |f|=1:}$

\begin{align}
3 \ \mathit{B_2} &= - 2 \,i\,\mathit{bd_1} - 2 \,i\,\mathit{ad_0}
- {\displaystyle
2\, \mathit{C_1}} ,\\
2\ \mathit{C_0} & = \mathit{a_1 }  -
\mathit{B_1}  - 2\ \mathit{bd_0} \,i,\\
2 \mathit{a_2} & = -B_2 ,\\
3\ \mathit{Bd_0} & = \,i\,\mathit{C_1} - \mathit{bd_1}  +
2\,\mathit{ad_0} ,\\
2\ \mathit{b_2} & = -  \mathit{a_1}  -
\mathit{B_1} ,\\
\mathit{b_1} & = - \mathit{B_0}.
\end{align}

\vskip .1cm

\noindent $\underline{\bullet \hbox{ For  } |f|=2:}$

\begin{align}
2\ \mathit{C_0} & = - 2\ \mathit{Bd_0}\,i - \mathit{B_1}  +
 \,i\,\mathit{ad_0} -  \,i\,
\mathit{bd_1} ,  \\
2\ \mathit{b_2} & = - \,i\,\mathit{ad_0} -  \, i\,\mathit{bd_1} -
\mathit{B_1} ,  \\
\mathit{bd_0} & =\mathit{b_1}\,i + \mathit{B_0}\,i.\label{exito}
\end{align}

\vskip .1cm

\noindent $\underline{\bullet \hbox{ For  } |f|=3:}$

\begin{align}
\mathit{C_0} & =\mathit{Cd_0}\,i,\\
\mathit{bd_0} &= - i \,\mathit{b_0},\\
\mathit{B_1} & =\mathit{Bd_1}\,i + \mathit{a_1} -
\mathit{ad_1}\,i, \\
\mathit{b_2} & =\mathit{bd_2}\,i - \mathit{a_1} +
\mathit{ad_1}\,i, \\
\mathit{bd_1} & = - \mathit{Bd_0} - \mathit{B_0}\,i -
\mathit{b_1}\,i,\\
\mathit{ad_0} & = - \mathit{a_0}\,i + \mathit{Bd_0} +
\mathit{B_0}\,i .
\end{align}

\vskip .1cm

\noindent $\underline{\bullet \hbox{ For  } f\in B_{\so(6)}:}$

\begin{align}
b_2 & =0,\\
b_1 & =0,\\
b_0 & =0.\label{ecborel3}
\end{align}
\end{lemma}

\

\begin{remark} The equations in Lemma \ref{conditions} are written using
the previously introduced notation. For example, strictly
speaking, if $\vec{m}=\p^2 m_2+\p m_1+m_0$ then   equation
(\ref{exito}) for an element $f=\xi_j\xi_k$ means
\begin{equation}\label{noo}
    bd(\xi_j\xi_k,m_0)=b(\xi_j\xi_k,m_1)\,i + B(\xi_j\xi_k,m_0)\,i.
\end{equation}
\end{remark}

\

\begin{proof}
Using this notation,  by taking coefficients in $\p^i \la^j$,
conditions (S1)-(S3) translate into the following list, for $f\in
\La (6)$ {\bf (see file "equations.mws" where the computations
were done using Maple, this file is located in the link written at
the beginning of the Appendices)}:

\vskip 0.3cm

\noindent $\underline{\bullet \hbox{ For  } |f|=0:}$

\vskip .3cm

$ \qquad \qquad\qquad \mathit{C_0}= - \mathit{B_1} -
\mathit{b_2},$

\vskip .3cm

$ \qquad \qquad \qquad \mathit{B_2}= - {\displaystyle \frac
{\mathit{a_2}}{2}}  - {\displaystyle \frac {\mathit{C_1}}{2}} ,$

\vskip .3cm

$ \qquad \qquad \qquad \mathit{bd_0}={\displaystyle \frac {1}{2}}
\,i\,\mathit{a_2} - {\displaystyle \frac {1}{2}} \,i
\,\mathit{C_1},$

\vskip .3cm

$ \qquad \qquad \qquad \mathit{ad_0}=\mathit{Bd_0},$

\vskip .3cm

$ \qquad \qquad \qquad \mathit{C_2}=0, \,\mathit{Cd_2}=0,
\,\mathit{ ad_2}=0, \,\mathit{Bd_2}=0, \,\mathit{Cd_1}=0,
\,\mathit{Cd_0}=0,$

\vskip .3cm

$ \qquad \qquad \qquad \mathit{bd_2}= - \mathit{Bd_1},$

\vskip .3cm

$ \qquad \qquad \qquad \mathit{bd_1}=
 - \mathit{Bd_0},$

\vskip .3cm

$
 \qquad \qquad \qquad \mathit{ad_1}=\mathit{Bd_1},$

\vskip .3cm

\noindent $\underline{\bullet \hbox{ For  } |f|=1:}$

\vskip .3cm

$ \qquad \qquad \qquad \mathit{B_2}= - {\displaystyle \frac
{2}{3}} \,i\,\mathit{bd_1} - {\displaystyle \frac {2}{3}}
\,i\,\mathit{ad_0} - {\displaystyle \frac {2\, \mathit{C_1}}{3}}
,$

\vskip .3cm

$\qquad \qquad \qquad \mathit{a_2}= {\displaystyle \frac {1}{3}}
\,i\,\mathit{bd_1} + {\displaystyle \frac {1}{3}}
\,i\,\mathit{ad_0} + {\displaystyle \frac {\mathit{ C_1}}{3}} ,$

\vskip .3cm

$ \qquad \qquad \qquad \mathit{C_0}={\displaystyle \frac
{\mathit{a_1 }}{2}}  - {\displaystyle \frac {\mathit{B_1}}{2}}  -
\mathit{bd_0} \,i,$

\vskip .3cm

$ \qquad \qquad \qquad \mathit{C_2}= - \mathit{Bd_1}\,i -
\mathit{bd_2}\,i, $

\vskip .3cm

$\qquad \qquad \qquad \mathit{Bd_0}={\displaystyle \frac {1}{3}}
\,i\,\mathit{C_1} - {\displaystyle \frac {\mathit{bd_1}}{3}}  +
{\displaystyle \frac { 2\,\mathit{ad_0}}{3}} ,$

\vskip .3cm

$\qquad \qquad \qquad  \mathit{b_2}= - {\displaystyle \frac {
\mathit{a_1}}{2}}  - {\displaystyle \frac {\mathit{B_1}}{2}} ,$

\vskip .3cm

$ \qquad \qquad \qquad \mathit{b_1}= - \mathit{B_0},$

\vskip .3cm

$ \qquad \qquad \qquad \mathit{Cd_2}=0, \mathit{ad_2}=0,
\mathit{Bd_2}=0, \mathit{Cd_1 }=0, \,\mathit{Cd_0}=0,$

\vskip .3cm

$ \qquad \qquad \qquad \mathit{ad_1}=\mathit{Bd_1},$

\vskip .3cm

\noindent $\underline{\bullet \hbox{ For  } |f|=2:}$

\vskip .3cm

$\qquad \qquad \qquad \mathit{B_2}= - \mathit{bd_2}\,i -
{\displaystyle \frac {1}{3}} \,i\,\mathit{ad_1}
 - {\displaystyle \frac {2\,\mathit{C_1}}{3}}  - {\displaystyle
\frac {2}{3}} \,i\,\mathit{Bd_1}, $

\vskip .3cm

$ \qquad \qquad \qquad \mathit{Cd_0}={\displaystyle \frac {1}{3}}
\,i\,\mathit{C_1} - {\displaystyle \frac {\mathit{Bd_1}}{3}}  +
{\displaystyle \frac { \mathit{ad_1}}{3}} , $

\vskip .3cm

$\qquad \qquad \qquad \mathit{C_0}= - \mathit{Bd_0}\,i -
{\displaystyle \frac {\mathit{B_1}}{2}}  + {\displaystyle \frac {1
}{2}} \,i\,\mathit{ad_0} - {\displaystyle \frac {1}{2}} \,i\,
\mathit{bd_1} + {\displaystyle \frac {\mathit{a_1}}{2}} ,  $

\vskip .3cm

$ \qquad \qquad \qquad \mathit{a_2}={\displaystyle \frac
{\mathit{C_1}}{3}}  + {\displaystyle \frac {1}{3}}
\,i\,\mathit{Bd_1} - {\displaystyle \frac {1}{3}}
\,i\,\mathit{ad_1},$

\vskip .3cm

$ \qquad \qquad \qquad \mathit{b_2}= - {\displaystyle \frac
{1}{2}} \,i\,\mathit{ad_0} - {\displaystyle \frac {1}{2}} \,
i\,\mathit{bd_1} - {\displaystyle \frac {\mathit{a_1}}{2}}  -
{\displaystyle \frac {\mathit{B_1}}{2}} ,  $

\vskip .3cm

$ \qquad \qquad \qquad \mathit{bd_0}=\mathit{b_1}\,i +
\mathit{B_0}\,i,$

\vskip .3cm

$ \qquad \qquad \qquad \mathit{C_2}= - i \,\mathit{Bd_2}, $

\vskip .3cm

$ \qquad \qquad \qquad \mathit{Cd_2}=0,   \mathit{ad_2}=0,
\,\mathit{Cd_1}=0, $

\vskip .3cm

\noindent $\underline{\bullet \hbox{ For  } |f|=3:}$

\vskip .3cm

$\qquad \qquad \qquad \mathit{C_0}=\mathit{Cd_0}\,i,$

\vskip .3cm

$ \qquad \qquad \qquad \mathit{C_1}= \mathit{Cd_1}\,i, $

\vskip .3cm

$ \qquad \qquad \qquad \mathit{C_2}=\mathit{Cd_2}\,i,$

\vskip .3cm

$ \qquad \qquad \qquad \mathit{bd_0}= - i \,\mathit{b_0},$

\vskip .3cm

$ \qquad \qquad \qquad \mathit{a_2}=\mathit{ad_2}\,i, $

\vskip .3cm

$ \qquad \qquad \qquad \mathit{B_1}=\mathit{Bd_1}\,i +
\mathit{a_1} - \mathit{ad_1}\,i, $

\vskip .3cm

$ \qquad \qquad \qquad \mathit{b_2}=\mathit{bd_2}\,i -
\mathit{a_1} + \mathit{ad_1}\,i, $

\vskip .3cm

$ \qquad \qquad \qquad \mathit{bd_1}= - \mathit{Bd_0} -
\mathit{B_0}\,i - \mathit{b_1}\,i,$

\vskip .3cm

$ \qquad \qquad \qquad \mathit{ad_0}= - \mathit{a_0}\,i +
\mathit{Bd_0} + \mathit{B_0}\,i , $

\vskip .3cm

 $
 \qquad \qquad \qquad \mathit{B_2}=\mathit{Bd_2}\,i,
 $

\vskip .3cm

\noindent $\underline{\bullet \hbox{ For  } f\in B_{\so(6)}:}$

\vskip .3cm

$ \qquad \qquad \qquad 0=b_0,$

$ \qquad \qquad \qquad 0=b_1+a_0,$

$ \qquad \qquad \qquad 0=b_2 + a_1,$

$ \qquad \qquad \qquad 0=a_2.$

\vskip .2cm

Now, taking care of the length of the elements $\xi_I$ involved in
the expression of vectors in Ind$(F)$ of degree at most -5, we
observe that some equations are always zero, getting the list in
the statement of the lemma.
\end{proof}

\

In order classify the singular vectors we should impose equations
(\ref{ec-a-1})-(\ref{ecborel3}) to the 5 possible forms of
singular vectors listed in (\ref{m-singular}), depending on the
degree. The following lemmas describe the result in each case.

\begin{lemma}\label{m5} All the singular vectors of degree -5 are listed in the theorem.
\end{lemma}

\begin{proof} Using the
softwares Macaulay2 and Maple, the conditions of Lemma
\ref{conditions} on the singular vector $\vec m_5$ were simplified
in several steps.  First, the conditions of Lemma \ref{conditions}
were reduced to a linear system of equations with a $992 \times
544$ matrix. After the reduction of this linear system, we
obtained in the middle of the file "m5-macaulay-2" a simplified
list of 542 equations (see Appendix \ref{appendix A} for the
details of this reduction). In particular, we obtained  the
following identities:

\begin{align}\label{m5-c}
& 0= \
 v_1+v_{1,3,4,5,6}
      & 0= \  v_2-v_{1,3,4,5,6}\ i                    \nonumber
      \\
      & 0= \  v_3+v_{1,2,3,5,6}
      & 0= \  v_4-v_{1,2,3,5,6}\ i                    \nonumber
      \\
      & 0= \  v_5+v_{1,2,3,4,5}
      & 0= \  v_6-v_{1,2,3,4,5}\ i                    \nonumber
      \\
      & 0= \  v_{1,2,3}-v_{1,2,3,5,6}\ i
      & 0= \  v_{1,2,4}-v_{1,2,3,5,6}               \nonumber  \\
      & 0= \  v_{1,2,5}-v_{1,2,3,4,5}\ i
      & 0= \  v_{1,2,6}-v_{1,2,3,4,5}               \nonumber  \\
      & 0= \  v_{1,3,4}-v_{1,3,4,5,6}\ i
      & 0= \  v_{1,3,5}+v_{2,4,6}\ i                  \nonumber
      \\
      & 0= \  v_{1,3,6}+v_{2,4,6}                  \ \
      & 0= \  v_{1,4,5}+v_{2,4,6}                     \\
      & 0= \  v_{1,4,6}-v_{2,4,6}\ i                 \ \
      & 0= \  v_{1,5,6}-v_{1,3,4,5,6}\ i             \nonumber
      \\
      & 0= \  v_{2,3,4}-v_{1,3,4,5,6}              \ \
      & 0= \  v_{2,3,5}+v_{2,4,6}                   \nonumber  \\
      & 0= \  v_{2,3,6}-v_{2,4,6}\ i                 \ \
      & 0= \  v_{2,4,5}-v_{2,4,6}\ i                  \nonumber
      \\
      & 0= \  v_{2,5,6}-v_{1,3,4,5,6}              \ \
      & 0= \  v_{3,4,5}-v_{1,2,3,4,5}\ i              \nonumber
      \\
      & 0= \  v_{3,4,6}-v_{1,2,3,4,5}              \ \
      & 0= \  v_{3,5,6}-v_{1,2,3,5,6}\ i              \nonumber
      \\
      & 0= \  v_{4,5,6}-v_{1,2,3,5,6}              \ \
      & 0= \  v_{2,3,4,5,6}+v_{1,3,4,5,6}\ i          \nonumber
      \\
      & 0= \  v_{1,2,4,5,6}+v_{1,2,3,5,6}\ i         \ \
      & 0= \  v_{1,2,3,4,6}+v_{1,2,3,4,5}\ i           \nonumber
\end{align}

\

\noindent In particular, all the vectors $v_I$ can be written in
terms of $v_1, v_3, v_5$ and $v_{1,3,5}$. By imposing this
identities, we obtained at the end of the file "m5-macaulay-2" the
following simplified list of 64 equations (see Appendix
\ref{appendix A} for the details of this reduction):\

\begin{eqnarray}\label{m5-d}
& 0=\  H_1v_1+1/2v_1        \qquad\qquad \      & 0=\
        H_2v_1-1/2v_1            \nonumber \\ & 0=\
        H_3v_1-1/2v_1         \qquad\qquad \       & 0=\
        H_1v_3-1/2v_3           \nonumber   \\ & 0=\
        H_2v_3+1/2v_3         \qquad\qquad     & 0=\
        H_3v_3-1/2v_3           \nonumber   \\ & 0=\
        H_1v_5-1/2v_5         \qquad\qquad \       & 0=\
        H_2v_5-1/2v_5            \nonumber   \\ & 0=\
        H_3v_5+1/2v_5          \qquad\qquad \     & 0=\
        H_1v_{1,3,5}+1/2v_{1,3,5}     \\ & 0=\
        H_2v_{1,3,5}+1/2v_{1,3,5} \qquad & 0=\
        H_3v_{1,3,5}+1/2v_{1,3,5}  \nonumber  \\
         & 0=\
        E_{00}v_1-9/2v_1        \qquad\qquad \    & 0=\
        E_{00}v_3-9/2v_3        \nonumber    \\ & 0=\
        E_{00}v_5-9/2v_5           \qquad\qquad    & 0=\
        E_{00}v_{1,3,5}-9/2v_{1,3,5}  \nonumber
\end{eqnarray}

\vskip .1cm

\noindent together with

\vskip .1cm

\begin{eqnarray}\label{m5-e}
         & 0=\
        E_{-(\ep_1-\ep_2)}v_1            \qquad \quad\qquad       & 0=\
        E_{-(\ep_1-\ep_2)}v_3-2v_1        \nonumber \\ & 0=\
        E_{-(\ep_1-\ep_2)}v_5             \qquad  \quad\qquad     & 0=\
        E_{-(\ep_1-\ep_2)}v_{1,3,5}       \nonumber  \\ & 0=\
        E_{-(\ep_1-\ep_3)}v_1            \qquad  \quad\qquad      & 0=\
        E_{-(\ep_1-\ep_3)}v_3            \nonumber  \\ & 0=\
        E_{-(\ep_1-\ep_3)}v_5-2v_1         \qquad      & 0=\
        E_{-(\ep_1-\ep_3)}v_{1,3,5}         \nonumber  \\ & 0=\
        E_{-(\ep_2-\ep_3)}v_1             \qquad  \quad\qquad     & 0=\
        E_{-(\ep_2-\ep_3)}v_3               \nonumber  \\ & 0=\
        E_{-(\ep_2-\ep_3)}v_5-2v_3        \qquad       & 0=\
        E_{-(\ep_2-\ep_3)}v_{1,3,5}         \nonumber \\ & 0=\
        E_{-(\ep_1+\ep_2)}v_1              \qquad \quad\qquad    & 0=\
        E_{-(\ep_1+\ep_2)}v_3              \\ & 0=\
        E_{-(\ep_1+\ep_2)}v_5-2v_{1,3,5}  \qquad     & 0=\
        E_{-(\ep_1+\ep_2)}v_{1,3,5}         \nonumber \\ & 0=\
        E_{-(\ep_1+\ep_3)}v_1             \qquad  \quad\qquad    & 0=\
        E_{-(\ep_1+\ep_3)}v_3+2v_{1,3,5}     \nonumber \\ & 0=\
        E_{-(\ep_1+\ep_3)}v_5             \qquad  \quad\qquad    & 0=\
        E_{-(\ep_1+\ep_3)}v_{1,3,5}          \nonumber  \\
        & 0=\
        E_{-(\ep_2+\ep_3)}v_1-2v_{1,3,5}   \qquad     & 0=\
        E_{-(\ep_2+\ep_3)}v_3               \nonumber \\ & 0=\
        E_{-(\ep_2+\ep_3)}v_5            \qquad  \quad\qquad     & 0=\
        E_{-(\ep_2+\ep_3)}v_{1,3,5}    \nonumber
\end{eqnarray}

\vskip .1cm

\noindent and

\vskip .1cm

\begin{eqnarray}\label{m5-f}
        & 0=\
        E_{\ep_1-\ep_2}v_1+2v_3         \qquad \qquad   & 0=\
        E_{\ep_1-\ep_2}v_3               \nonumber  \\ & 0=\
        E_{\ep_1-\ep_2}v_5         \quad\qquad  \qquad \qquad       & 0=\
        E_{\ep_1-\ep_2}v_{1,3,5}         \nonumber  \\ & 0=\
        E_{\ep_1-\ep_3}v_1+2v_5         \qquad \qquad   & 0=\
        E_{\ep_1-\ep_3}v_3                \nonumber \\ & 0=\
        E_{\ep_1-\ep_3}v_5         \quad\qquad \qquad \qquad        & 0=\
        E_{\ep_1-\ep_3}v_{1,3,5}          \nonumber \\ & 0=\
        E_{\ep_2-\ep_3}v_1        \quad\qquad  \qquad \qquad        & 0=\
        E_{\ep_2-\ep_3}v_3+2v_5           \nonumber \\ & 0=\
        E_{\ep_2-\ep_3}v_5         \quad\qquad \qquad \qquad        & 0=\
        E_{\ep_2-\ep_3}v_{1,3,5}         \nonumber  \\ & 0=\
        E_{\ep_1+\ep_2}v_1        \quad\qquad  \qquad \qquad        & 0=\
        E_{\ep_1+\ep_2}v_3               \\ & 0=\
        E_{\ep_1+\ep_2}v_5        \quad\qquad \qquad \qquad        & 0=\
        E_{\ep_1+\ep_2}v_{1,3,5}+2v_5    \nonumber \\ & 0=\
        E_{\ep_1+\ep_3}v_1        \quad\qquad \qquad \qquad        & 0=\
        E_{\ep_1+\ep_3}v_3               \nonumber \\ & 0=\
        E_{\ep_1+\ep_3}v_5     \quad\qquad   \qquad \qquad         & 0=\
        E_{\ep_1+\ep_3}v_{1,3,5}-2v_3    \nonumber \\ & 0=\
        E_{\ep_2+\ep_3}v_1         \quad\qquad \qquad \qquad       & 0=\
        E_{\ep_2+\ep_3}v_3               \nonumber \\ & 0=\
        E_{\ep_2+\ep_3}v_5       \quad\qquad \qquad \qquad          & 0=\
        E_{\ep_2+\ep_3}v_{1,3,5}+2v_1    \nonumber
\end{eqnarray}

\

\noindent Therefore, a  vector
$$
\vec m_5=\p^2 \ \sum_{|I|= 5}
\xi_I\otimes v_{I}\ +\ \p \ \sum_{|I|= 3} \xi_I\otimes v_{I}
+\sum_{|I|= 1} \xi_I\otimes v_{I}
$$
satisfies  conditions of Lemma \ref{conditions} if and only if
equations (\ref{m5-c}-\ref{m5-f}) hold. We divided the final
analysis of these equations in several cases:

\vskip .1cm

\noindent $\bullet$ Case $v_5\neq 0$: Using (\ref{m5-f}) we obtain
that the Borel subalgebra of $\so(6)$ annihilates $v_5$. Hence, it
is a highest weight vector in the irreducible $\cso(6)$-module
$F_\mu$, and by (\ref{m5-d}), the highest weight is

\begin{equation}\label{m5-mu}
\mu=\left(\dfrac{9}{2} \
;\dfrac{1}{2},\dfrac{1}{2},-\dfrac{1}{2}\right).
\end{equation}
Using (\ref{m5-e}), the other vectors are completely determined by
the highest weight vector $v_5$, namely
\begin{align}\label{m5-v}
 v_1 & =\frac{1}{2} E_{-(\ep_1-\ep_3)}\ v_5, \nonumber \\
 v_3 & =\frac{1}{2} E_{-(\ep_2-\ep_3)}\ v_5, \\
 v_{1,3,5} & =\frac{1}{2} E_{-(\ep_1+\ep_2)}\ v_5. \nonumber
\end{align}
\vskip .1cm
\noindent After lengthly computation, it is easy to see that
(\ref{m5-d}-\ref{m5-f}) hold by using (\ref{m5-mu}) and
(\ref{m5-v}). Therefore, the vector
\begin{align}\label{m5-m}
\vec m_5  =   & \p^2  \left[\sum_{l=1}^3 (\xi_{\{2l\}^c}-i\
\xi_{\{2l-1\}^c})\otimes v_{2l-1} \right] \nonumber \\
& +\ \p  \bigg[ \ (i\ \xi_{134} + \xi_{234}+ i \xi_{156} +
\xi_{256})\otimes
v_1   \\
& \quad\  \ + \ (i\ \xi_{123} + \xi_{124}+ i \xi_{356} +
\xi_{456})\otimes
v_3  \nonumber\\
& \quad \ \ + \ (i\ \xi_{125} + \xi_{126}+ i \xi_{345} +
\xi_{346})\otimes
v_5 \nonumber \\
&  \quad\ \  + \ (i\ \xi_{136} + \xi_{236}+ i \xi_{145} +
\xi_{146}+ i\ \xi_{235} + \xi_{245}- i \xi_{246} -
\xi_{135})\otimes
v_{1,3,5} \bigg] \nonumber \\
& + \  \sum_{l=1}^3 (\xi_{2l}+i\ \xi_{2l-1})\otimes v_{2l-1}
\nonumber
\end{align}
is a singular vector of Ind$(F_\mu)$, where $v_5$ is a highest
weight vector in $F_\mu$, $\mu=(9/2;1/2,1/2,-1/2)$ and
$v_1,v_3,v_{1,3,5}$ are given by (\ref{m5-v}). By computing
\begin{align}\label{wt-sing-vect}
E_{00} \cdot \vec m_5 & =\hbox{coefficient of  } \la^1
\left(1\,_{\la} \vec m_5\right)\nonumber \\
H_1 \cdot \vec m_5 & =\hbox{coefficient of  } \la^0
\left(-i\ \xi_1\xi_2 \,_{\la} \vec m_5\right)  \\
H_2 \cdot \vec m_5 & =\hbox{coefficient of  } \la^0
\left(-i\ \xi_3\xi_4 \,_{\la} \vec m_5\right) \nonumber \\
H_3 \cdot \vec m_5 & =\hbox{coefficient of  } \la^0 \left(-i\
\xi_5\xi_6 \,_{\la} \vec m_5\right),\nonumber
\end{align}
one can prove that
\begin{equation*}
    wt\ \vec m_5=\left(\dfrac{9}{2} \
;\dfrac{1}{2},\dfrac{1}{2},\dfrac{1}{2}\right),
\end{equation*}
finishing this case.

\vskip .2cm

\noindent $\bullet$ Case $v_5= 0$ and $v_3\neq 0$: In this case,
using (\ref{m5-f}) and (\ref{m5-d}), we have that $v_3$ is a
highest weight vector in $F_\mu$, with $\mu=(9/2;1/2,-1/2,1/2)$
that is not dominant integral, getting a contradiction.

\vskip .2cm

\noindent $\bullet$ Case $v_5= 0$, $v_3= 0$ and $v_1\neq 0$: In
this case, using (\ref{m5-f}) and (\ref{m5-d}), we have that $v_1$
is a highest weight vector in $F_\mu$, with
$\mu=(9/2;-1/2,1/2,1/2)$ that is not dominant integral, getting a
contradiction.

\vskip .2cm

\noindent $\bullet$ Case $0=v_5 =v_3=v_1$ and $v_{1,3,5}\neq 0$:
In this case, using (\ref{m5-f}) and (\ref{m5-d}), we have that
$v_{1,3,5}$ is a highest weight vector in $F_\mu$, with
$\mu=(9/2;-1/2,-1/2,-1/2)$ that is not  dominant integral, getting
a contradiction and finishing the proof.
\end{proof}

\

\begin{lemma}\label{m4} There is no singular vector of degree -4.
\end{lemma}

\begin{proof} The proof of this lemma was done entirely with the
softwares Macaulay2 and Maple. The conditions on the singular
vector $\vec m_4$ were reduced to a linear systems of equations
with  a $1104 \times 527$ matrix, whose rank is 527 (see Appendix
\ref{appendix A} for details). Therefore, there is no non-trivial
solution of this linear system, proving that there is no singular
vector of degree -4, finishing the lemma.
\end{proof}

\begin{lemma}\label{5.5} All the singular vectors of degree -3 are listed in the theorem.
\end{lemma}

\begin{proof} Using the
softwares Macaulay2 and Maple, the conditions of Lemma
\ref{conditions} on the singular vector $\vec m_3$ were simplified
in several steps.  First, the conditions of Lemma \ref{conditions}
were reduced to a linear system of equations with a $694 \times
442$ matrix. After the reduction of this linear system, we
obtained at the end of the file "m3-macaulay-1" a simplified list
of 397 equations (see Appendix \ref{appendix A} for the details of
this reduction). In particular, we obtained  the following
identities:

\begin{align}\label{m3-cond-copy}
                   & 0 =  \ v_{1,2,3}-v_{4,5,6}\ i
       & 0 =  \  v_{1,2,4}+v_{3,5,6}\ i
      \nonumber \\ & 0 =  \  v_{1,2,5}-v_{3,4,6}\ i
        & 0 =  \  v_{1,2,6}+v_{3,4,5}\ i
      \nonumber \\ & 0 =  \  v_{1,3,4}-v_{2,5,6}\ i
        & 0 =  \  v_{1,3,5}+v_{2,4,6}\ i
      \nonumber \\ & 0 =  \  v_{1,3,6}-v_{2,4,5}\ i
        & 0 =  \  v_{1,4,5}-v_{2,3,6}\ i
      \nonumber \\ & 0 =  \  v_{1,4,6}+v_{2,3,5}\ i
        & 0 =  \  v_{1,5,6}-v_{2,3,4}\ i
                \\ & 0 =  \  v_{2,3,4,5,6}
        & 0 =  \  v_{1,3,4,5,6} \hskip .88cm \
      \nonumber \\ & 0 =  \  v_{1,2,4,5,6}
        & 0 =  \  v_{1,2,3,5,6} \hskip .88cm \
      \nonumber \\ & 0 =  \  v_{1,2,3,4,6}
        & 0 =  \  v_{1,2,3,4,5} \hskip .88cm \ \nonumber
 \end{align}

\

\noindent Now, we have to impose the identities
(\ref{m3-cond-copy}) to reduce the number of variables. Observe
that  everything can be written in terms of
\begin{equation*}\label{1111}
    v_{1,j,k} \ \ \hbox{ with } \ 2\leq j<k\leq 6.
\end{equation*}
Unfortunately, the result is not enough to obtain in a clear way
the possible highest weight vectors. For example, after the
reduction and some extra computations it is possible to see that
\begin{equation}\label{hhh}
(v_{1,3,6}-v_{1,4,5})+\ i\ (v_{1,3,5}+v_{1,4,6})
\end{equation}
is  annihilated by the Borel subalgebra. Hence, it is necessary to
impose (\ref{m3-cond-copy}) and  make a change of variables. We
produced an auxiliary file where we imposed (\ref{m3-cond-copy}),
and after the analysis of the results, we found that the following
change of variable is convenient:

\begin{align}\label{ui-copy}
u_1 & = v_{1,2,3}- i\ v_{1,2,4} \nonumber \\
u_2 & = v_{1,2,3}+ i\ v_{1,2,4} \nonumber \\
u_3 & = v_{1,2,5}- i\ v_{1,2,6} \nonumber \\
u_4 & = v_{1,2,6}+ i\ v_{1,2,6} \nonumber \\
u_5 & = v_{1,3,4}- \ v_{1,5,6}  \\
u_6 & = v_{1,3,4}+ \ v_{1,5,6} \nonumber \\
u_7 & = v_{1,3,5}-v_{1,4,6}+ i\ (v_{1,3,6}+v_{1,4,5}) \nonumber \\
u_8 & = v_{1,3,5}+v_{1,4,6}- i\ (v_{1,3,6}-v_{1,4,5}) \nonumber \\
u_9 & = v_{1,3,5}-v_{1,4,6}- i\ (v_{1,3,6}+v_{1,4,5}) \nonumber \\
u_{10} & = v_{1,3,5}+v_{1,4,6}+ i\ (v_{1,3,6}-v_{1,4,5}) \nonumber
\end{align}

\vskip .2cm

\noindent Observe that all the equations  will be written in terms
of $u_i$ with $1\leq i\leq 10$. By imposing this identities, we
obtained at the end of the file "m3-macaulay-2" the following
simplified list of 125 equations  (see Appendix \ref{appendix A}
for the details of this reduction):

\begin{align}
 & 0= \  H_1u_1-1/4E_{-(\varepsilon_1-\varepsilon_3)}u_{10}\ i\ +u_1  \label{C-1}         \\ & 0=\
       -H_1u_1+H_2u_1-u_1                             \label{C-2}   \\ & 0=\
       H_1u_1+H_3u_1                                  \label{C-3}   \\ & 0=\
        H_1u_2-1/4E_{-(\varepsilon_1+\varepsilon_2)}u_5-1/4E_{-(\varepsilon_1+\varepsilon_3)}u_8\ i\ +3/2u_2 \label{C-4}   \\ & 0=\
       H_1u_2+H_2u_2-1/2E_{-(\varepsilon_1+\varepsilon_2)}u_5+2u_2             \label{C-5}   \\ & 0=\
       -H_1u_2+H_3u_2+1/2E_{-(\varepsilon_1+\varepsilon_2)}u_5-u_2             \label{C-6}   \\ & 0=\
        H_1u_3-1/4E_{-(\varepsilon_1-\varepsilon_3)}u_5+1/4E_{-(\varepsilon_1-\varepsilon_2)}u_8\ i\ +3/2u_3   \label{C-7}   \\ & 0=\
       H_1u_3+H_2u_3-1/2E_{-(\varepsilon_1-\varepsilon_3)}u_5+u_3               \label{C-8}   \\ & 0=\
       -H_1u_3+H_3u_3+1/2E_{-(\varepsilon_1-\varepsilon_3)}u_5-2u_3             \label{C-9}   \\ & 0=\
       H_1u_4+H_3u_4+u_4                              \label{C-10}   \\ & 0=\
       H_2u_4+H_3u_4+u_4                              \label{C-11}   \\ & 0=\
        H_3u_4-1/4E_{-(\varepsilon_1+\varepsilon_2)}u_{10}\ i\                        \label{C-12}   \\ & 0=\
        H_1u_5+1/4E_{-(\varepsilon_2-\varepsilon_3)}u_{10}\ i\                         \label{C-13}   \\ & 0=\
       -H_1u_5+H_2u_5+u_5                             \label{C-14}   \\ & 0=\
       H_2u_5+H_3u_5                                  \label{C-15}   \\ & 0=\
       1/2E_{-(\varepsilon_1+\varepsilon_3)}u_3+H_1u_6+H_3u_6+2u_6             \label{C-16}   \\ & 0=\
       1/2E_{-(\varepsilon_1+\varepsilon_2)}u_1+H_1u_6+H_2u_6+2u_6             \label{C-17}   \\ & 0=\
        H_2u_6+H_3u_6+1/2E_{-(\varepsilon_2+\varepsilon_3)}u_9\ i\ +2u_6            \label{C-18}   \\ & 0=\
        -E_{-(\varepsilon_1+\varepsilon_2)}u_4\ i\ +H_1u_7+H_2u_7-H_3u_7+E_{00}u_7-2u_7 \label{C-19}   \\ & 0=\
        E_{-(\varepsilon_1+\varepsilon_3)}u_2\ i\ +H_1u_7-H_2u_7+H_3u_7+E_{00}u_7-2u_7  \label{C-20}   \\ & 0=\
        -E_{-(\varepsilon_2+\varepsilon_3)}u_6\ i\ -H_1u_7+H_2u_7+H_3u_7+E_{00}u_7-2u_7 \label{C-21}
\end{align}

%\vfill

%\pagebreak

        \begin{align}
         & 0=\
       H_1u_8-E_{00}u_8+4u_8                              \label{C-22}   \\ & 0=\
        -1/2E_{-(\varepsilon_2-\varepsilon_3)}u_5\ i\ +H_2u_8+E_{00}u_8-3u_8             \label{C-23}   \\ & 0=\
       H_2u_8+H_3u_8                                  \label{C-24}   \\ & 0=\
       H_1u_9+H_2u_9+H_3u_9-E_{00}u_9+4u_9                \label{C-25}   \\ & 0=\
        1/2E_{-(\varepsilon_1-\varepsilon_2)}u_3\ i\ +H_2u_9-E_{00}u_9+3u_9              \label{C-26}   \\ & 0=\
        -1/2E_{-(\varepsilon_1-\varepsilon_3)}u_1\ i\ +H_3u_9-E_{00}u_9+3u_9             \label{C-27}   \\ & 0=\
       H_1u_{10}-E_{00}u_{10}+4u_{10}                           \label{C-28}   \\ & 0=\
       H_2u_{10}-E_{00}u_{10}+4u_{10}                           \label{C-29}   \\ & 0=\
       H_3u_{10}+E_{00}u_{10}-4u_{10}                           \label{C-30}   \\ & 0=\
       -H_1u_1+E_{00}u_1-5u_1                             \label{C-31}   \\ & 0=\
       -H_1u_2+E_{00}u_2-5u_2                             \label{C-32}   \\ & 0=\
       -H_1u_3+E_{00}u_3-5u_3                             \label{C-33}   \\ & 0=\
       H_3u_4+E_{00}u_4-4u_4                              \label{C-34}
 \\
        & 0=\
        -H_3u_5+E_{00}u_5+1/2E_{-(\varepsilon_2-\varepsilon_3)}u_{10}\ i\ -3u_5            \label{C-35}   \\ & 0=\
        -H_1u_6+E_{00}u_6+1/2E_{-(\varepsilon_2+\varepsilon_3)}u_9\ i\
        -4u_6          \label{C-36}
        \end{align}
together with
        \begin{align}
 & 0=\
       E_{-(\varepsilon_1-\varepsilon_2)}u_1                                    \label{C-37}   \\ & 0=\
       E_{-(\varepsilon_2-\varepsilon_3)}u_1+E_{-(\varepsilon_1-\varepsilon_3)}u_5                        \label{C-38}   \\ & 0=\
       E_{-(\varepsilon_1-\varepsilon_2)}u_2+E_{-(\varepsilon_1+\varepsilon_3)}u_3                       \label{C-39}   \\ & 0=\
       E_{-(\varepsilon_1-\varepsilon_3)}u_2-E_{-(\varepsilon_1+\varepsilon_2)}u_3                       \label{C-40}   \\ & 0=\
        E_{-(\varepsilon_2-\varepsilon_3)}u_2+E_{-(\varepsilon_1+\varepsilon_2)}u_8\ i\                       \label{C-41}   \\ & 0=\
        E_{-(\varepsilon_1-\varepsilon_3)}u_3+E_{-(\varepsilon_2-\varepsilon_3)}u_9\ i\                        \label{C-42}   \\ & 0=\
        E_{-(\varepsilon_2-\varepsilon_3)}u_3+E_{-(\varepsilon_1-\varepsilon_3)}u_8\ i\                        \label{C-43}   \\ & 0=\
       E_{-(\varepsilon_1-\varepsilon_2)}u_4                                    \label{C-44}   \\ & 0=\
       E_{-(\varepsilon_1+\varepsilon_2)}u_1+E_{-(\varepsilon_1-\varepsilon_3)}u_4                       \label{C-45}   \\ & 0=\
       E_{-(\varepsilon_2-\varepsilon_3)}u_4-E_{-(\varepsilon_1+\varepsilon_2)}u_5                       \label{C-46}   \\ & 0=\
       E_{-(\varepsilon_1-\varepsilon_2)}u_5+2u_1                               \label{C-47}   \\ & 0=\
        E_{-(\varepsilon_1-\varepsilon_2)}u_6-E_{-(\varepsilon_1+\varepsilon_3)}u_9\ i\                       \label{C-48}   \\ & 0=\
        E_{-(\varepsilon_1-\varepsilon_3)}u_6+E_{-(\varepsilon_1+\varepsilon_2)}u_9\ i\                       \label{C-49}   \\ & 0=\
       -E_{-(\varepsilon_1+\varepsilon_2)}u_3+E_{-(\varepsilon_2-\varepsilon_3)}u_6                      \label{C-50}   \\ & 0=\
        E_{-(\varepsilon_1+\varepsilon_3)}u_6\ i\ +E_{-(\varepsilon_1-\varepsilon_2)}u_7                      \label{C-51}   \\ & 0=\
        -E_{-(\varepsilon_1+\varepsilon_2)}u_6\ i\ +E_{-(\varepsilon_1-\varepsilon_3)}u_7                     \label{C-52}   \\ & 0=\
        -E_{-(\varepsilon_1+\varepsilon_2)}u_2\ i\ +E_{-(\varepsilon_2-\varepsilon_3)}u_7                     \label{C-53}
\\
         & 0=\
       E_{-(\varepsilon_1-\varepsilon_2)}u_{10}                                   \label{C-54}   \\ & 0=\
       E_{-(\varepsilon_1+\varepsilon_3)}u_1                                   \label{C-55}   \\ & 0=\
       E_{-(\varepsilon_2+\varepsilon_3)}u_1+2u_4                              \label{C-56}
   \end{align}
        \begin{align}
        & 0=\
       E_{-(\varepsilon_2+\varepsilon_3)}u_2                                   \label{C-57}   \\ & 0=\
       E_{-(\varepsilon_2+\varepsilon_3)}u_3-2u_2                              \label{C-58}
\\
        & 0=\
       E_{-(\varepsilon_1+\varepsilon_3)}u_4                                   \label{C-59}   \\ & 0=\
       E_{-(\varepsilon_2+\varepsilon_3)}u_4                                   \label{C-60}   \\ & 0=\
       E_{-(\varepsilon_1+\varepsilon_3)}u_5+2u_4                              \label{C-61}   \\ & 0=\
       E_{-(\varepsilon_2+\varepsilon_3)}u_5                                   \label{C-62}   \\ & 0=\
       E_{-(\varepsilon_2+\varepsilon_3)}u_8                                   \label{C-63}   \\ & 0=\
       E_{-(\varepsilon_1+\varepsilon_3)}u_{10}                                  \label{C-64}   \\ & 0=\
       E_{-(\varepsilon_2+\varepsilon_3)}u_{10}                  \label{C-65}
\\
        & 0=\
       E_{\varepsilon_1-\varepsilon_2}u_1-2u_5                                \label{C-66}   \\ & 0=\
        E_{\varepsilon_1-\varepsilon_3}u_1+2u_{10}\ i\                               \label{C-67}   \\ & 0=\
       E_{\varepsilon_2-\varepsilon_3}u_1                                     \label{C-68}   \\ & 0=\
       E_{\varepsilon_1-\varepsilon_2}u_2                                     \label{C-69}   \\ & 0=\
       E_{\varepsilon_1-\varepsilon_3}u_2                                     \label{C-70}   \\ & 0=\
       E_{\varepsilon_2-\varepsilon_3}u_2+2u_4                                \label{C-71}   \\ & 0=\
        E_{\varepsilon_1-\varepsilon_2}u_3-2u_8\ i\                                \label{C-72}   \\ & 0=\
       E_{\varepsilon_1-\varepsilon_3}u_3+2u_5                                \label{C-73}   \\ & 0=\
       E_{\varepsilon_2-\varepsilon_3}u_3-2u_1                                \label{C-74}   \\ & 0=\
       E_{\varepsilon_1-\varepsilon_2}u_4                                     \label{C-75}   \\ & 0=\
       E_{\varepsilon_1-\varepsilon_3}u_4                                     \label{C-76}   \\ & 0=\
       E_{\varepsilon_2-\varepsilon_3}u_4                                     \label{C-77}   \\ & 0=\
       E_{\varepsilon_1-\varepsilon_2}u_5                                     \label{C-78}   \\ & 0=\
       E_{\varepsilon_1-\varepsilon_3}u_5                                     \label{C-79}   \\ & 0=\
        E_{\varepsilon_2-\varepsilon_3}u_5-2u_{10}\ i\                               \label{C-80}   \\ & 0=\
       E_{\varepsilon_1-\varepsilon_2}u_6+2u_2                                \label{C-81}   \\ & 0=\
       E_{\varepsilon_1-\varepsilon_3}u_6+2u_4                                \label{C-82}
\\
        & 0=\
       E_{\varepsilon_2-\varepsilon_3}u_6                                     \label{C-83}   \\ & 0=\
       E_{\varepsilon_1-\varepsilon_2}u_7                                     \label{C-84}   \\ & 0=\
       E_{\varepsilon_1-\varepsilon_3}u_7                                     \label{C-85}   \\ & 0=\
       E_{\varepsilon_2-\varepsilon_3}u_7                                     \label{C-86}   \\ & 0=\
       E_{\varepsilon_1-\varepsilon_2}u_8                                     \label{C-87}   \\ & 0=\
       E_{\varepsilon_1-\varepsilon_3}u_8                                     \label{C-88}   \\ & 0=\
        E_{\varepsilon_2-\varepsilon_3}u_8+4u_5\ i\                                \label{C-89}   \\ & 0=\
        E_{\varepsilon_1-\varepsilon_2}u_9+4u_3\ i\                                \label{C-90}   \\ & 0=\
        E_{\varepsilon_1-\varepsilon_3}u_9-4u_1\ i\                                \label{C-91}   \\ & 0=\
       E_{\varepsilon_2-\varepsilon_3}u_9                                     \label{C-92}   \\ & 0=\
       E_{\varepsilon_1-\varepsilon_2}u_{10}                                    \label{C-93}   \\ & 0=\
       E_{\varepsilon_1-\varepsilon_3}u_{10}                                    \label{C-94}   \\ & 0=\
       E_{\varepsilon_2-\varepsilon_3}u_{10}                 \label{C-95}
\end{align}

%\vfill

%\pagebreak

        \begin{align}
        & 0=\
       E_{\varepsilon_1+\varepsilon_2}u_1                                    \label{C-96}   \\ & 0=\
       E_{\varepsilon_1+\varepsilon_3}u_1                                    \label{C-97}   \\ & 0=\
       E_{\varepsilon_2+\varepsilon_3}u_1                                    \label{C-98}   \\ & 0=\
       E_{\varepsilon_1+\varepsilon_2}u_2+2u_5                               \label{C-99}   \\ & 0=\
        E_{\varepsilon_1+\varepsilon_3}u_2+2u_8\ i\                               \label{C-100}   \\ & 0=\
       E_{\varepsilon_2+\varepsilon_3}u_2+2u_3                               \label{C-101}   \\ & 0=\
       E_{\varepsilon_1+\varepsilon_2}u_3                                    \label{C-102}   \\ & 0=\
       E_{\varepsilon_1+\varepsilon_3}u_3                                    \label{C-103}   \\ & 0=\
       E_{\varepsilon_2+\varepsilon_3}u_3                                    \label{C-104}   \\ & 0=\
        E_{\varepsilon_1+\varepsilon_2}u_4-2u_{10}\ i\                              \label{C-105}   \\ & 0=\
       E_{\varepsilon_1+\varepsilon_3}u_4-2u_5                               \label{C-106}   \\ & 0=\
       E_{\varepsilon_2+\varepsilon_3}u_4-2u_1                               \label{C-107}   \\ & 0=\
       E_{\varepsilon_1+\varepsilon_2}u_5                                    \label{C-108}   \\ & 0=\
       E_{\varepsilon_1+\varepsilon_3}u_5                                    \label{C-109}   \\ & 0=\
       E_{\varepsilon_2+\varepsilon_3}u_5                                    \label{C-110}   \\ & 0=\
       E_{\varepsilon_1+\varepsilon_2}u_6-2u_1                               \label{C-111}   \\ & 0=\
       E_{\varepsilon_1+\varepsilon_3}u_6-2u_3                               \label{C-112}   \\ & 0=\
        E_{\varepsilon_2+\varepsilon_3}u_6-2u_9\ i\                               \label{C-113}   \\ & 0=\
        E_{\varepsilon_1+\varepsilon_2}u_7+4u_4\ i\                               \label{C-114}   \\ & 0=\
        E_{\varepsilon_1+\varepsilon_3}u_7-4u_2\ i\                               \label{C-115}   \\ & 0=\
        E_{\varepsilon_2+\varepsilon_3}u_7+4u_6\ i\                               \label{C-116}   \\ & 0=\
       E_{\varepsilon_1+\varepsilon_2}u_8                                    \label{C-117}   \\ & 0=\
       E_{\varepsilon_1+\varepsilon_3}u_8                                    \label{C-118}   \\ & 0=\
       E_{\varepsilon_2+\varepsilon_3}u_8                                    \label{C-119}   \\ & 0=\
       E_{\varepsilon_1+\varepsilon_2}u_9                                    \label{C-120}   \\ & 0=\
       E_{\varepsilon_1+\varepsilon_3}u_9                                    \label{C-121}   \\ & 0=\
       E_{\varepsilon_2+\varepsilon_3}u_9                                    \label{C-122}   \\ & 0=\
       E_{\varepsilon_1+\varepsilon_2}u_{10}                                   \label{C-123}   \\ & 0=\
       E_{\varepsilon_1+\varepsilon_3}u_{10}                                   \label{C-124}   \\ & 0=\
       E_{\varepsilon_2+\varepsilon_3}u_{10}               \label{C-125}
      \end{align}

\

Therefore, a  singular vector  of degree -3 must have the
simplified form
$$
\vec m_3= \sum_{|I|= 3} \xi_I\otimes v_{I}
$$
and it satisfies  conditions of Lemma \ref{conditions} if and only
if equations (\ref{m3-cond-copy}), (\ref{ui-copy}) and
(\ref{C-1}-\ref{C-125}) hold. We divided the final analysis of
these equations in several cases:

\vskip .1cm

\noindent $\bullet$ Case $u_{10}\neq 0$: Using
(\ref{C-93}-\ref{C-95}) and  (\ref{C-123}-\ref{C-125}) we obtain
that the Borel subalgebra of $\so(6)$ annihilates $u_{10}$. Hence,
it is a highest weight vector in the irreducible $\cso(6)$-module
$F_\mu$, and by (\ref{C-28}-\ref{C-30}), the highest weight is

\begin{equation}\label{m3-mu}
\mu=\left(k+4 \ ; k , k ,-k\right), \qquad \hbox{with } 2k\in
\ZZ_{\geq 0}.
\end{equation}

\vskip .1cm

\noindent Then we shall prove that the cases $k=0$ and $k=1/2$ are not
possible. Using (\ref{C-67}), (\ref{C-83}), (\ref{C-105}) and
other similar equations, we deduce that if $u_{10}\neq 0$ then
$u_i\neq 0$ for all $i$. Now, we shall see that all $u_i$  are
completely determined by the highest weight vector $u_{10}$. If $k=0$, we are working with the trivial $\so (6)$ representation, and using (\ref{C-67}) we obtain $u_{10}=0$ getting  a contradiction. Assume that $k\neq 0$. Now, applying $E_{\ep_1-\ep_3}$ to (\ref{C-1}), we can prove
that
\begin{equation}\label{m3-u1}
u_1=\frac{i}{4k} E_{-(\ep_1-\ep_3)} u_{10}.
\end{equation}
Similarly, using (\ref{C-12}) and (\ref{C-16}), we have
\begin{align}\label{m3-u4}
u_4=-\frac{i}{4k} E_{-(\ep_1+\ep_2)} u_{10},\\
u_5=-\frac{i}{4k} E_{-(\ep_2-\ep_3)} u_{10}.\label{m3-u5}
\end{align}
Applying $E_{\ep_1+\ep_2}$ to (\ref{C-5}) and using (\ref{C-99}),
we can prove that
\begin{equation}\label{u222}
2(2k-1)u_2=E_{-(\ep_1+\ep_2)} u_{5}.
\end{equation}
If {\bf $k\neq \frac{1}{2}$}, we have
\begin{align}\label{m3-u2}
 u_2 & =\frac{1}{2(2k-1)} E_{-(\ep_1+\ep_2)}\ u_5 \nonumber \\
& =\frac{-i}{8k(2k-1)} E_{-(\ep_1+\ep_2)}E_{-(\ep_2-\ep_3)}\ u_{10} \\
& =\frac{1}{2(2k-1)} E_{-(\ep_2-\ep_3)}\ u_4. \nonumber
\end{align}

\vskip .1cm

\noindent If $k=\frac{1}{2}$, we are working with a spin
representation. Using (\ref{u222}) we get $0=E_{-(\ep_1+\ep_2)}
u_{5}$. In this case, by (\ref{C-5}-\ref{C-6}) and (\ref{C-32}),
we have wt
$u_2=(\dfrac{1}{2}+4;-\dfrac{1}{2},-\dfrac{3}{2},\dfrac{1}{2})$,
which is impossible in a spin representation (see p. 288
\cite{Knapp}).

\

\noindent Similarly, using (\ref{C-8}) and (\ref{C-73}), we have

\vskip .1cm

\begin{align}\label{m3-u3}
 u_3 & =\frac{1}{2(2k-1)} E_{-(\ep_1-\ep_3)}\ u_5 \nonumber \\
& =\frac{-i}{8k(2k-1)} E_{-(\ep_1-\ep_3)}E_{-(\ep_2-\ep_3)}\ u_{10} \\
& =\frac{-1}{2(2k-1)} E_{-(\ep_2-\ep_3)}\ u_1. \nonumber
\end{align}

\vskip .1cm

\noindent Using (\ref{C-14}) and (\ref{C-108}), we have

\vskip .1cm

\begin{align}\label{m3-u6}
 u_6 & =\frac{-1}{2(2k-1)} E_{-(\ep_1+\ep_2)}\ u_1 \nonumber \\
& =\frac{-i}{8k(2k-1)} E_{-(\ep_1+\ep_2)}E_{-(\ep_1-\ep_3)}\ u_{10} \\
& =\frac{1}{2(2k-1)} E_{-(\ep_1-\ep_3)}\ u_4. \nonumber
\end{align}

\vskip .1cm

\noindent By (\ref{C-19}), (\ref{C-105}), (\ref{C-114}) and
(\ref{m3-u4}), we have

\vskip .1cm

\begin{align}\label{m3-u7}
 u_7 & =\frac{i}{2(2k-1)} E_{-(\ep_1+\ep_2)}\ u_4 \nonumber \\
& =\frac{1}{8k(2k-1)} E_{-(\ep_1+\ep_2)}E_{-(\ep_1+\ep_2)}\ u_{10}
.
\end{align}

\vskip .1cm

\noindent Using (\ref{C-23}), (\ref{C-83}), (\ref{C-89}) and
(\ref{m3-u5}), we have

\vskip .1cm

\begin{align}\label{m3-u8}
 u_8 & =\frac{i}{2(2k-1)} E_{-(\ep_2-\ep_3)}\ u_5 \nonumber \\
& =\frac{1}{8k(2k-1)} E_{-(\ep_2-\ep_3)}E_{-(\ep_2-\ep_3)}\ u_{10}
.
\end{align}

\vskip .1cm

\noindent By (\ref{C-27}), (\ref{C-67}), (\ref{C-91}) and
(\ref{m3-u1}), we have

\vskip .1cm

\begin{align}\label{m3-u9}
 u_9 & =\frac{-i}{2(2k-1)} E_{-(\ep_1-\ep_3)}\ u_1 \nonumber \\
& =\frac{1}{8k(2k-1)} E_{-(\ep_1-\ep_3)}E_{-(\ep_1-\ep_3)}\ u_{10}
.
\end{align}

\vskip .1cm

After a lengthly computation, it is possible to see that
(\ref{C-1}-\ref{C-125}) hold by using (\ref{m3-mu}) with $k\neq
1/2$, and the expressions of $u_i$ obtained in
(\ref{m3-u1}-\ref{m3-u9}). Therefore, using the expressions of
$v_{k,l,j}$'s given in terms of $u_i$ as in (\ref{ch-var}), the
vector
\begin{align}\label{m3-m}
\vec m_3  & =   2 \left[(\xi_{\{1,2,3\}}- \ i  \
\xi_{\{1,2,3\}^c})-(\xi_{\{3,5,6\}}- \ i  \
\xi_{\{3,5,6\}^c})\right]\otimes u_1,\nonumber\\
 & + 2 \left[(\xi_{\{1,2,3\}}- \ i  \
\xi_{\{1,2,3\}^c})+(\xi_{\{3,5,6\}}- \ i  \
\xi_{\{3,5,6\}^c})\right]\otimes u_2,\nonumber\\
& + 2 \left[(\xi_{\{1,2,5\}}- \ i  \
\xi_{\{1,2,5\}^c})-(\xi_{\{3,4,5\}}- \ i  \
\xi_{\{3,4,5\}^c})\right]\otimes u_3,\nonumber\\
& + 2 \left[(\xi_{\{1,2,5\}}- \ i  \
\xi_{\{1,2,5\}^c})+(\xi_{\{3,4,5\}}- \ i  \
\xi_{\{3,4,5\}^c})\right]\otimes u_4,  \\
& + 2 \left[(\xi_{\{1,3,4\}}- \ i  \
\xi_{\{1,3,4\}^c})-(\xi_{\{1,5,6\}}- \ i  \
\xi_{\{1,5,6\}^c})\right]\otimes u_5,\nonumber\\
& + 2 \left[(\xi_{\{1,3,4\}}- \ i  \
\xi_{\{1,3,4\}^c})+(\xi_{\{1,5,6\}}- \ i  \
\xi_{\{1,5,6\}^c})\right]\otimes u_6,\nonumber\\
& +  \left[(\xi_{\{1,3,5\}}+ \ i  \
\xi_{\{1,3,5\}^c})-(\xi_{\{2,4,5\}}+ \ i  \
\xi_{\{2,4,5\}^c})\right. \nonumber \\
&\ \ \ \  \left. -(\xi_{\{2,3,6\}}+ \ i  \
\xi_{\{2,3,6\}^c})-(\xi_{\{1,4,6\}}+ \ i  \
\xi_{\{1,4,6\}^c})\right]\otimes u_7,\nonumber\\
& +  \left[(\xi_{\{1,3,5\}}+ \ i  \
\xi_{\{1,3,5\}^c})+(\xi_{\{2,4,5\}}+ \ i  \
\xi_{\{2,4,5\}^c})\right. \nonumber \\
&\ \ \ \  \left. -(\xi_{\{2,3,6\}}+ \ i  \
\xi_{\{2,3,6\}^c})+(\xi_{\{1,4,6\}}+ \ i  \
\xi_{\{1,4,6\}^c})\right]\otimes u_8,\nonumber\\
& +  \left[(\xi_{\{1,3,5\}}+ \ i  \
\xi_{\{1,3,5\}^c})+(\xi_{\{2,4,5\}}+ \ i  \
\xi_{\{2,4,5\}^c})\right. \nonumber \\
&\ \ \ \  \left. +(\xi_{\{2,3,6\}}+ \ i  \
\xi_{\{2,3,6\}^c})-(\xi_{\{1,4,6\}}+ \ i  \
\xi_{\{1,4,6\}^c})\right]\otimes u_9,\nonumber\\
& +  \left[(\xi_{\{1,3,5\}}+ \ i  \
\xi_{\{1,3,5\}^c})-(\xi_{\{2,4,5\}}+ \ i  \
\xi_{\{2,4,5\}^c})\right. \nonumber \\
&\ \ \ \  \left. +(\xi_{\{2,3,6\}}+ \ i  \
\xi_{\{2,3,6\}^c})+(\xi_{\{1,4,6\}}+ \ i  \
\xi_{\{1,4,6\}^c})\right]\otimes u_{10},\nonumber\\
\end{align}
is a singular vector of Ind$(F_\mu)$, where the   $u_i$'s are
written in  (\ref{m3-u1}-\ref{m3-u9}) in terms of $u_{10}$, where
$u_{10}$ is a highest weight vector in $F_\mu$, and
$\mu=(k+4;k,k,-k)$ with $2k\in \ZZ_{> 0}$ and $k\neq \frac{1}{2}$.
Now using (\ref{wt-sing-vect}), one can prove that
\begin{equation*}
wt\ \vec m_3=\left(k+1 \ ; k-1 , k-1 ,-k+1\right)
\end{equation*}
finishing this case.

\vskip .2cm

\noindent $\bullet$ Case $u_{10}= 0$ and $u_5\neq 0$: In this
case, using (\ref{C-78}-\ref{C-80}) and (\ref{C-108}-\ref{C-110}),
we have that $u_5$ is a highest weight vector in $F_\mu$.
Considering (\ref{C-13}-\ref{C-15}) and (\ref{C-35}), we have
$\mu=(4;0,-1,1)$ that is not dominant integral, getting a
contradiction.

\vskip .2cm

\noindent $\bullet$ Case $u_{10}=u_5= 0$ and $u_1\neq 0$: In this
case, using (\ref{C-66}-\ref{C-68}) and (\ref{C-96}-\ref{C-98}),
we have that $u_1$ is a highest weight vector in $F_\mu$.
Considering (\ref{C-1}-\ref{C-3}) and (\ref{C-31}), we have
$\mu=(4;-1,0,1)$ that is not dominant integral, getting a
contradiction.

\vskip .2cm

\noindent $\bullet$ Case $u_{10}=u_5=u_1= 0$ and $u_8\neq 0$: In
this case, using (\ref{C-87}-\ref{C-89}) and
(\ref{C-117}-\ref{C-119}), we have that $u_8$ is a highest weight
vector in $F_\mu$. Considering (\ref{C-22}-\ref{C-24}) we have
$\mu=(k+4;k,-k-1,k+1)$ that is not dominant integral, getting a
contradiction.

\vskip .2cm

\noindent $\bullet$ Case $u_{10}=u_5=u_1=u_8= 0$ and $u_4\neq 0$:
In this case, using (\ref{C-75}-\ref{C-77}) and
(\ref{C-105}-\ref{C-107}), we have that $u_4$ is a highest weight
vector in $F_\mu$. Considering (\ref{C-10}-\ref{C-12})  and
(\ref{C-34}), we have $\mu=(4;-1,-1,0)$ that is not dominant
integral, getting a contradiction.

\vskip .2cm

\noindent $\bullet$ Case $u_{10}=u_5=u_1=u_8=u_4= 0$ and $u_3\neq
0$: In this case, using (\ref{C-72}-\ref{C-74}) and
(\ref{C-102}-\ref{C-104}), we have that $u_3$ is a highest weight
vector in $F_\mu$. Considering (\ref{C-7}-\ref{C-9})  and
(\ref{C-33}), we have $\mu=(7/2;-3/2,1/2,1/2)$ that is not
dominant integral, getting a contradiction.

\vskip .2cm

\noindent $\bullet$ Case $u_{10}=u_5=u_1=u_8=u_4=u_3= 0$ and
$u_2\neq 0$: In this case, using (\ref{C-69}-\ref{C-71}) and
(\ref{C-99}-\ref{C-101}), we have that $u_2$ is a highest weight
vector in $F_\mu$. Considering (\ref{C-4}-\ref{C-6})  and
(\ref{C-32}), we have $\mu=(7/2;-3/2,-1/2,-1/2)$ that is not
dominant integral, getting a contradiction.

\vskip .2cm

\noindent $\bullet$ Case $u_{10}=u_5=u_1=u_8=u_4=u_3=u_2= 0$ and
$u_9\neq 0$: In this case, using (\ref{C-90}-\ref{C-92}) and
(\ref{C-120}-\ref{C-122}), we have that $u_9$ is a highest weight
vector in $F_\mu$. Considering (\ref{C-25}-\ref{C-27}), we have
$\mu=(-k+2;k,-(k+1),-(k+1))$ that is not dominant integral,
getting a contradiction.

\vskip .2cm

\noindent $\bullet$ Case $u_{10}=u_5=u_1=u_8=u_4=u_3=u_2=u_9= 0$
and $u_6\neq 0$: In this case, using (\ref{C-81}-\ref{C-83}) and
(\ref{C-111}-\ref{C-113}), we have that $u_6$ is a highest weight
vector in $F_\mu$. Considering (\ref{C-16}-\ref{C-18})  and
(\ref{C-36}), we have $\mu=(3;-1,-1,-1)$ that is not dominant
integral, getting a contradiction.

\vskip .2cm

\noindent $\bullet$ Case $u_{10}=u_5=u_1=u_8=u_4=u_3=u_2=u_9=u_6=
0$ and $u_7\neq 0$: In this case, using (\ref{C-84}-\ref{C-86})
and (\ref{C-114}-\ref{C-116}), we have that $u_7$ is a highest
weight vector in $F_\mu$. Considering (\ref{C-19}-\ref{C-21}), we
have $\mu=(-k+2;k,k,k)$ which is a multiple of the spin
representation with $2k\in \ZZ_{\geq 0}$. In this case $u_i=0$ for all
$i\neq 7$ and most of the equations (\ref{C-1}-\ref{C-125}) are
trivial, and it is easy to check that the remaining equations hold
in this case. Therefore, using (\ref{ch-var}), we have that the
vector
\begin{align}\label{m3-sing-2}
    \vec m_3= & \left( \xi_{\{1,3,5\}}-\xi_{\{1,4,6\}}-\ i \
    (\xi_{\{1,3,6\}}+\xi_{\{1,4,5\}})\right)\otimes u_7  \nonumber \\
& +\  \left( \ i\ (\xi_{\{1,3,5\}^c}-\xi_{\{1,4,6\}^c})- \
    (\xi_{\{1,3,6\}^c}+\xi_{\{1,4,5\}^c})\right)\otimes u_7
\end{align}
is a singular vector in Ind$(F_\mu)$, where $\mu=(-k+2;k,k,k)$
 with $2k\in \ZZ_{\geq 0}$ and $u_7$ is a highest
weight vector in $F_\mu$. Now using (\ref{wt-sing-vect}), one can
prove that
\begin{equation*}
wt\ \vec m_3=\left(-k-1 \ ; k+1 , k+1 ,k+1\right)
\end{equation*}
finishing the classification of singular vectors of degree -3.
\end{proof}

\vskip .4cm

\begin{lemma}\label{m2} There is no singular vector of degree  -2.
\end{lemma}

\vskip .2cm

\begin{proof} Using the
softwares Macaulay2 and Maple, the conditions of Lemma
\ref{conditions} on the singular vector $\vec m_2$ were reduced to
a linear system of equations with a $268 \times 272$ matrix. After
the reduction of this linear system, we obtained at the end of the
file "m2-macaulay" or in file "m2-ecuations.pdf" a simplified list
of 192 equations (see Appendix \ref{appendix A} for the details of
this reduction). The 15th and 16th equations of this list are the
following
\begin{align}
\ & \ 0 =\
-i*F_{1,2}*v_{3,4,5,6}+E*v_{3,4,5,6}-5*v_{3,4,5,6}+i*v_{1,2,3,4,5,6}\label{m2-1}\\
\ & \ 0 =\ E*v_{1,2,3,4,5,6}-3*v_{1,2,3,4,5,6}\label{m2-2}
\end{align}
and at the end of this list we have the conditions
\begin{align}\label{m2-3}
\ & \ 0 =\ F_{1,2}*v_{1,2,3,4,5,6}-v_{3,4,5,6}\nonumber\\ \ & \ 0
=\ F_{1,3}*v_{1,2,3,4,5,6}+v_{2,4,5,6}\nonumber\\ \ & \ 0 =\
F_{1,4}*v_{1,2,3,4,5,6}-v_{2,3,5,6}\nonumber\\ \ & \ 0 =\
F_{1,5}*v_{1,2,3,4,5,6}+v_{2,3,4,6}\nonumber\\ \ & \ 0 =\
F_{1,6}*v_{1,2,3,4,5,6}-v_{2,3,4,5}\nonumber\\ \ & \ 0 =\
F_{2,3}*v_{1,2,3,4,5,6}-i*v_{2,4,5,6}\nonumber\\ \ & \ 0 =\
F_{2,4}*v_{1,2,3,4,5,6}+i*v_{2,3,5,6}\nonumber\\ \ & \ 0 =\
F_{2,5}*v_{1,2,3,4,5,6}-i*v_{2,3,4,6}\\ \ & \ 0 =\
F_{2,6}*v_{1,2,3,4,5,6}+i*v_{2,3,4,5}\nonumber\\ \ & \ 0 =\
F_{3,4}*v_{1,2,3,4,5,6}-v_{1,2,5,6}\nonumber\\ \ & \ 0 =\
F_{3,5}*v_{1,2,3,4,5,6}+v_{1,2,4,6}\nonumber\\ \ & \ 0 =\
F_{3,6}*v_{1,2,3,4,5,6}-v_{1,2,4,5}\nonumber\\ \ & \ 0 =\
F_{4,5}*v_{1,2,3,4,5,6}-i*v_{1,2,4,6}\nonumber\\ \ & \ 0 =\
F_{4,6}*v_{1,2,3,4,5,6}+i*v_{1,2,4,5}\nonumber\\ \ & \ 0 =\
F_{5,6}*v_{1,2,3,4,5,6}-v_{1,2,3,4}.\nonumber
\end{align}

\

\noindent Therefore, if $\vec m_2 =\p \ \sum_{|I|= 6} \xi_I\otimes
v_{I} +\sum_{|I|= 4} \xi_I\otimes v_{I}$ is a singular vector in
Ind$(F_\mu)$, using equations (\ref{m2-3}), we prove that
$v_{1,2,3,4,5,6}\in F_\mu$ is annihilated by the Borel subalgebra
of $\so(6)$ (see (\ref{eq:borel})), and using that $F_\mu$ is
irreducible, we get that $v_{1,2,3,4,5,6}$ is a highest weight
vector. Now, we shall compute the corresponding weight $\mu$.
Recall (\ref{H}) and observe that using  (\ref{m2-2}) and
(\ref{m2-3}), the equation (\ref{m2-1}) is equivalent to the
following
\begin{equation}\label{m2-4}
    (H_1^2+2H_1+1)v_{1,2,3,4,5,6}=0,
\end{equation}
obtaining that $H_1v_{1,2,3,4,5,6}=-v_{1,2,3,4,5,6}$. Therefore,
the weight $\mu$ is not dominant integral, getting a contradiction
and finishing the proof.
\end{proof}

\

\begin{lemma}\label{m1} All the singular vectors of degree -1 are listed in the theorem.
\end{lemma}

\begin{proof} Since the singular vectors found in \cite{BKL} for $K_6$ are also
singular vectors for $CK_6$,  using (B42-B43) in \cite{BKL}, we
have that it is convenient to introduce the following notation:
\begin{align}\label{vect-m1-bis}
\vec{m}_1 & =\sum_{i=1}^6 \xi_{\{i\}^c}\otimes v_{\{i\}^c}\\
&
=\sum_{l=1}^3\bigg[\big(\xi_{\{2l\}^c}+i\xi_{\{2l-1\}^c}\big)\otimes
w_l + \big(\xi_{\{2l\}^c}-i \xi_{\{2l-1\}^c}\big)\otimes
\overline{w}_l\bigg]\nonumber
\end{align}
that is, for $1\leq l\leq 3$
\begin{equation}\label{v-w-bis}
    v_{\{2l\}^c}  =w_l+\overline{w}_l,\qquad
    v_{\{2l-1\}^c} =i(w_l-\overline{w}_l)
\end{equation}
or equivalently, for $1\leq l\leq 3$
\begin{equation}\label{w-v-bis}
    w_{l}  =\frac{1}{2}(v_{\{2l\}^c}-i\ v_{\{2l-1\}^c}),\qquad
    \overline{w}_{l} =\frac{1}{2}(v_{\{2l\}^c}+i\ v_{\{2l-1\}^c}).
\end{equation}

We applied the change of variables (\ref{v-w-bis}), and using the
softwares Macaulay2 and Maple, the conditions of Lemma
\ref{conditions} on the singular vector $\vec m_1$ were simplified
in several steps. First, the conditions of Lemma \ref{conditions}
were reduced to a linear system of equations with a $62 \times
102$ matrix. After the reduction of this linear system, we
obtained at the end of the file "m1-macaulay" a simplified list of
51 equations (see Appendix \ref{appendix A} for the details of
this reduction). More precisely, we obtained  the following
identities:

\begin{align}
           &   0 = \   H_1 w_1-E_{00}w_1+4w_1
     \label{m1-1}   \\   &   0 = \   H_2 w_1+H_3 w_1
     \label{m1-2}   \\   &   0 = \   \frac{1}{2}E_{-(\ep_1-\ep_2)}w_1+H_1 w_2+H_3 w_2+w_2
     \label{m1-3}   \\   &   0 = \   -\frac{1}{2}E_{-(\ep_1-\ep_2)}w_1+H_2 w_2-E_{00}w_2+3w_2
     \label{m1-4}   \\   &   0 = \   \frac{1}{2}E_{-(\ep_1-\ep_3)}w_1+\frac{1}{2}E_{-(\ep_2-\ep_3)}w_2+H_1 w_3+H_2 w_3+2w_3
     \label{m1-5}   \\   &   0 = \   -\frac{1}{2}E_{-(\ep_1-\ep_3)}w_1-\frac{1}{2}E_{-(\ep_2-\ep_3)}w_2+H_3 w_3-E_{00}w_3+2w_3
     \label{m1-6}   \\   &   \hbox{\tiny $ 0 = \   \frac{1}{2}E_{-(\ep_1+\ep_2)}w_2+\frac{1}{2}E_{-(\ep_1+\ep_3)}w_3+H_1
      \overline{w}_1+E_{00}\overline{w}_1-\frac{1}{2}E_{-(\ep_1-\ep_2)}\overline{w}_2-\frac{1}{2}E_{-(\ep_1-\ep_3)}\overline{w}_3$}
     \label{m1-7}   \\   &   \hbox{\tiny $ 0 = \   \frac{1}{2}E_{-(\ep_1+\ep_2)}w_2-\frac{1}{2}E_{-(\ep_1+\ep_3)}w_3+H_2 \overline{w}_1-H_3
      \overline{w}_1+\frac{1}{2}E_{-(\ep_1-\ep_2)}\overline{w}_2-\frac{1}{2}E_{-(\ep_1-\ep_3)}\overline{w}_3
      $}
     \label{m1-8}   \\   &   \hbox{\tiny $ 0 = \   -\frac{1}{2}E_{-(\ep_1+\ep_2)}w_1-\frac{1}{2}E_{-(\ep_2+\ep_3)}w_3+H_1 \overline{w}_2-H_3
      \overline{w}_2-\frac{1}{2}E_{-(\ep_2-\ep_3)}\overline{w}_3+\overline{w}_2$}
     \label{m1-9}   \\   &   \hbox{\tiny $ 0 = \   -\frac{1}{2}E_{-(\ep_1+\ep_2)}w_1+\frac{1}{2}E_{-(\ep_2+\ep_3)}w_3+H_2
      \overline{w}_2+E_{00}\overline{w}_2-\frac{1}{2}E_{-(\ep_2-\ep_3)}\overline{w}_3-\overline{w}_2 $}
     \label{m1-10}   \\   &   0 = \   -\frac{1}{2}E_{-(\ep_1+\ep_3)}w_1+\frac{1}{2}E_{-(\ep_2+\ep_3)}w_2+H_1 \overline{w}_3-H_2 \overline{w}_3
     \label{m1-11}   \\   &   0 = \   -\frac{1}{2}E_{-(\ep_1+\ep_3)}w_1-\frac{1}{2}E_{-(\ep_2+\ep_3)}w_2+H_3 \overline{w}_3+E_{00}\overline{w}_3-2\overline{w}_3
     \label{m1-12}   \\   &   0 = \   -E_{-(\ep_1+\ep_2)}w_3+E_{-(\ep_2-\ep_3)}\overline{w}_1-E_{-(\ep_1-\ep_3)}\overline{w}_2
     \label{m1-13}   \\   &   0 = \   E_{-(\ep_1+\ep_3)}w_2+E_{-(\ep_1-\ep_2)}\overline{w}_3
     \label{m1-14}   \\   &   0 = \   E_{-(\ep_2+\ep_3)}w_1 \label{m1-15}
\end{align}
and
\begin{align}
           &   0 = \   E_{\ep_1-\ep_2}w_1
     \label{m1-16}   \\   &   0 = \   E_{\ep_1-\ep_3}w_1
     \label{m1-17}   \\   &   0 = \   E_{\ep_2-\ep_3}w_1
     \label{m1-18}   \\   &   0 = \   E_{\ep_1-\ep_2}w_2-2w_1
     \label{m1-19}   \\   &   0 = \   E_{\ep_1-\ep_3}w_2
     \label{m1-20}   \\   &   0 = \   E_{\ep_2-\ep_3}w_2
     \label{m1-21}   \\   &   0 = \   E_{\ep_1-\ep_2}w_3
     \label{m1-22}   \\   &   0 = \   E_{\ep_1-\ep_3}w_3-2w_1
     \label{m1-23}   \\   &   0 = \   E_{\ep_2-\ep_3}w_3-2w_2
     \label{m1-24}   \\
     &   0 = \   E_{\ep_1-\ep_2}\overline{w}_1+2\overline{w}_2
     \label{m1-25}   \\   &   0 = \   E_{\ep_1-\ep_3}\overline{w}_1+2\overline{w}_3
     \label{m1-26}   \\   &   0 = \   E_{\ep_2-\ep_3}\overline{w}_1
     \label{m1-27}   \\   &   0 = \   E_{\ep_1-\ep_2}\overline{w}_2
     \label{m1-28}   \\   &   0 = \   E_{\ep_1-\ep_3}\overline{w}_2
     \label{m1-29}   \\   &   0 = \   E_{\ep_2-\ep_3}\overline{w}_2+2\overline{w}_3
     \label{m1-30}   \\   &   0 = \   E_{\ep_1-\ep_2}\overline{w}_3
     \label{m1-31}   \\   &   0 = \   E_{\ep_1-\ep_3}\overline{w}_3
     \label{m1-32}   \\   &   0 = \   E_{\ep_2-\ep_3}\overline{w}_3
     \label{m1-33}
\\
                          &   0 = \   E_{\ep_1+\ep_2}w_1
     \label{m1-34}   \\   &   0 = \   E_{\ep_1+\ep_3}w_1
     \label{m1-35}   \\   &   0 = \   E_{\ep_2+\ep_3}w_1
     \label{m1-36}   \\   &   0 = \   E_{\ep_1+\ep_2}w_2
     \label{m1-37}   \\   &   0 = \   E_{\ep_1+\ep_3}w_2
     \label{m1-38}   \\   &   0 = \   E_{\ep_2+\ep_3}w_2
     \label{m1-39}   \\   &   0 = \   E_{\ep_1+\ep_2}w_3
     \label{m1-40}   \\   &   0 = \   E_{\ep_1+\ep_3}w_3
     \label{m1-41}   \\   &   0 = \   E_{\ep_2+\ep_3}w_3
     \label{m1-42}   \\
        &   0 = \   E_{\ep_1+\ep_2}\overline{w}_1-2w_2
     \label{m1-43}   \\   &   0 = \   E_{\ep_1+\ep_3}\overline{w}_1-2w_3
     \label{m1-44}   \\   &   0 = \   E_{\ep_2+\ep_3}\overline{w}_1
     \label{m1-45}   \\   &   0 = \   E_{\ep_1+\ep_2}\overline{w}_2+2w_1
     \label{m1-46}   \\   &   0 = \   E_{\ep_1+\ep_3}\overline{w}_2
     \label{m1-47}   \\   &   0 = \   E_{\ep_2+\ep_3}\overline{w}_2-2w_3
     \label{m1-48}   \\   &   0 = \   E_{\ep_1+\ep_2}\overline{w}_3
     \label{m1-49}   \\   &   0 = \   E_{\ep_1+\ep_3}\overline{w}_3+2w_1
     \label{m1-50}   \\   &   0 = \
     E_{\ep_2+\ep_3}\overline{w}_3+2w_2 \label{m1-51}
\end{align}

We divided the final analysis of these equations in several cases:

\vskip .1cm

\noindent $\bullet$ Case $w_{1}\neq 0$: Using
(\ref{m1-16}-\ref{m1-18}) and  (\ref{m1-34}-\ref{m1-36}) we obtain
that the Borel subalgebra of $\so(6)$ annihilates $w_{1}$. Hence,
it is a highest weight vector in the irreducible $\cso(6)$-module
$F_\mu$, and by (\ref{m1-1}-\ref{m1-2}), the (dominant integral)
highest weight is

\begin{equation}\label{m1-mu-1}
\mu=\left(k+4 \ ; k , l ,-l\right), \qquad \hbox{with } 2k\in
\ZZ_{\geq 0}, 2l\in \ZZ_{\geq 0} \hbox{ and  } k-l\in \ZZ_{\geq
0}.
\end{equation}

\vskip .1cm

\noindent Then we shall prove that the case $k=l$ is not possible.
Using (\ref{m1-19}), (\ref{m1-23}), (\ref{m1-43}) and other
similar equations, we deduce that if $w_{1}\neq 0$ then $w_i\neq
0\neq \overline{w}_i$ for all $i$. Now, we shall see that all
$w_i$'s and $\overline{w}_i$'s are completely determined by the
highest weight vector $w_{1}$. More precisely, using
(\ref{m1-18}), we have that wt$\,_{w_2}=(k+4;k-1,l+1,-l)$. Hence,
from (\ref{m1-3}) we can prove that
\begin{equation}\label{m1-w2-no}
2(l-k)w_2= E_{-(\ep_1-\ep_2)} w_{1}.
\end{equation}
If {\bf $k\neq l$}, we have
\begin{equation}\label{m1-w2}
w_2=\frac{1}{2(l-k)} E_{-(\ep_1-\ep_2)} w_{1}.
\end{equation}
If {\bf $k= l$}, using (\ref{m1-18}), we have
$w_2\in[F_\mu]_{\mu-(\ep_1-\ep_2)}$ that has dimension 0 if $k=l$,
which is a contradiction since $w_2\neq 0$.

Therefore, from now on we shall assume that $k\neq l$. Similarly,
using (\ref{m1-23}), we have that wt$\,_{w_3}=(k+4;k-1,l,-l+1)$.
Hence, from  (\ref{m1-5}) we can prove that
\begin{align}\label{m1-w3}
w_3=\frac{-1}{2(k+l+1)} \left(E_{-(\ep_1-\ep_3)}
w_{1}+E_{-(\ep_2-\ep_3)} w_{2}\right).
\end{align}
Using (\ref{m1-50}), we have that
wt$\,_{\overline{w}_3}=(k+4;k-1,l,-l-1)$. Hence, from the sum of
(\ref{m1-11}) and (\ref{m1-12}) we can prove that
\begin{equation}\label{m1-lw3}
\overline{w}_3=\frac{1}{2(k-l)} E_{-(\ep_1+\ep_3)} w_{1}.
\end{equation}
Using (\ref{m1-46}), we have that
wt$\,_{\overline{w}_2}=(k+4;k-1,l-1,-l)$. Hence, from the sum of
(\ref{m1-9}) and (\ref{m1-10}) we can prove that
\begin{equation}\label{m1-lw2}
\overline{w}_2=\frac{1}{2(k+l+1)} \left(E_{-(\ep_1+\ep_2)}
w_{1}+E_{-(\ep_2-\ep_3)} \overline{w}_{3}\right).
\end{equation}
Using (\ref{m1-43}), we have that
wt$\,_{\overline{w}_1}=(k+4;k-2,l,-l)$. Hence, from the sum of
(\ref{m1-7}) and (\ref{m1-8}) we can prove that
\begin{equation}\label{m1-lw1}
\overline{w}_1=\frac{1}{2(k+l+1)} \left(E_{-(\ep_1-\ep_3)}
\overline{w}_{3}-E_{-(\ep_1+\ep_2)} {w}_{2}\right).
\end{equation}

Now, we have an explicit expression of all $w_i$'s and
$\overline{w}_j$'s in terms of $w_1$. After some lengthly
computations it is possible to prove that equations
(\ref{m1-1}-\ref{m1-51}) hold. Hence, the vector
\begin{align}\label{m1-sing-1}
    \vec{m}_1=\sum_{l=1}^3\bigg[\big(\xi_{\{2l\}^c}+i\xi_{\{2l-1\}^c}\big)\otimes
w_l + \big(\xi_{\{2l\}^c}-i \xi_{\{2l-1\}^c}\big)\otimes
\overline{w}_l\bigg]
\end{align}
is a singular vector, where ${w}_1$ is a highest weight vector of
$F_\mu$,  $\mu=\left(k+4 \ ; k , l ,-l\right)$, with
 $2k\in \ZZ_{\geq 0}, 2l\in \ZZ_{\geq 0},  k-l\in
\ZZ_{> 0}$, and all $w_i$'s and $\overline{w}_j$'s are written in
terms of $w_1$ in (\ref{m1-w2}), (\ref{m1-w3}), (\ref{m1-lw1}),
(\ref{m1-lw2}) and (\ref{m1-lw3}). Now using (\ref{wt-sing-vect}),
one can prove that
\begin{equation*}
wt\ \vec m_1=\left(k+3 \ ; k-1 , l ,-l\right)
\end{equation*}
finishing this case.

\vskip .5cm

\noindent $\bullet$ Case $w_1=0$ and $w_{2}\neq 0$: Using
(\ref{m1-19}-\ref{m1-21}) and  (\ref{m1-37}-\ref{m1-39}) we obtain
that the Borel subalgebra of $\so(6)$ annihilates $w_{2}$. Hence,
it is a highest weight vector in the irreducible $\cso(6)$-module
$F_\mu$, and considering (\ref{m1-3}-\ref{m1-4}), we have

\begin{equation*}
\mu=\left(l+3 \ ; k , l ,-k-1\right), \qquad \hbox{with } 2k\in
\ZZ_{\geq 0}, 2l\in \ZZ_{\geq 0},
\end{equation*}
that is not dominant integral, getting a contradiction.

\vskip .1cm

\noindent $\bullet$ Case $w_1=w_2=0$ and $w_{3}\neq 0$: Using
(\ref{m1-22}-\ref{m1-24}) and  (\ref{m1-40}-\ref{m1-42}) we obtain
that the Borel subalgebra of $\so(6)$ annihilates $w_{3}$. Hence,
it is a highest weight vector in the irreducible $\cso(6)$-module
$F_\mu$, and considering (\ref{m1-5}-\ref{m1-6}), we have

\begin{equation*}
\mu=\left(l+2 \ ; k , -k-2 ,l\right),
\end{equation*}
that is not dominant integral, getting a contradiction.

\vskip .1cm

\noindent $\bullet$ Case $w_1=w_2=w_3=0$ and $\overline{w}_{3}\neq
0$: Using (\ref{m1-31}-\ref{m1-33}) and  (\ref{m1-49}-\ref{m1-51})
we obtain that the Borel subalgebra of $\so(6)$ annihilates
$\overline{w}_{3}$. Hence, it is a highest weight vector in the
irreducible $\cso(6)$-module $F_\mu$, and considering
(\ref{m1-11}-\ref{m1-12}), we have

\begin{equation}\label{m1-mu-2}
\mu=\left(-l+2 \ ; k , k ,l\right), \qquad \hbox{with } 2k\in
\ZZ_{\geq 0}, 2l\in \ZZ, k+l\in \ZZ_{\geq 0}, k-l\in \ZZ_{\geq 0}.
\end{equation}
Then we will see that the case $k=l$ is not possible. Using
(\ref{m1-26}) and (\ref{m1-30}), we have $\overline{w}_{1}\neq
0\neq \overline{w}_{2}$.

Now, we shall see that all $\overline{w}_i$'s  are completely
determined by the highest weight vector $\overline{w}_{3}$. More
precisely, applying $E_{\ep_2-\ep_3}$ to (\ref{m1-9}), we can
prove that
\begin{equation}\label{m1-line-w2-no}
2(k-l)\overline{w}_2= E_{-(\ep_2-\ep_3)} \overline{w}_{3}.
\end{equation}
If {\bf $k\neq l$}, we have
\begin{equation}\label{m1-line-w2}
\overline{w}_2=\frac{1}{2(k-l)} E_{-(\ep_2-\ep_3)}
\overline{w}_{3}.
\end{equation}
If {\bf $k= l$}, using (\ref{m1-30}), we have $0\neq
\overline{w}_2\in [F_\mu]_{\mu-(\ep_2-\ep_3)}$, but dim
$[F_\mu]_{\mu-(\ep_2-\ep_3)}=0$, getting a contradiction.

Therefore, from now on we shall assume that $k\neq l$. Similarly,
applying $E_{\ep_1-\ep_3}$ to  the sum of (\ref{m1-7}) and
(\ref{m1-8}), we can prove that
\begin{align}\label{m1-line-w1}
\overline{w}_1=\frac{1}{2(k-l)} E_{-(\ep_1-\ep_3)}
\overline{w}_{3}.
\end{align}
In this case, equations (\ref{m1-1}-\ref{m1-51}) collapse to a few
ones and it is easy to see that all of them hold. Hence, the
vector
\begin{align}\label{m1-sing-2}
    \vec{m}_1= & \frac{1}{2(k-l)}\big(\xi_{\{2\}^c}-i \xi_{\{1\}^c}\big)\otimes
E_{-(\ep_1-\ep_3)} \overline{w}_{3}\ +\  \\
& + \frac{1}{2(k-l)}\big(\xi_{\{4\}^c}-i \xi_{\{3\}^c}\big)\otimes
E_{-(\ep_2-\ep_3)} \overline{w}_{3}+ \big(\xi_{\{6\}^c}-i
\xi_{\{5\}^c}\big)\otimes \overline{w}_{3} \nonumber
\end{align}
is a singular vector, where $\overline{w}_3$ is a highest weight
vector of $F_\mu$, and $\mu=\left(-l+2 \ ; k , k ,l\right)$, with
 $2k\in \ZZ_{\geq 0}, 2l\in \ZZ, k+l\in \ZZ_{\geq 0}, k-l\in
\ZZ_{> 0}$. Now using (\ref{wt-sing-vect}), one can prove that
\begin{equation*}
wt\ \vec m_1=\left(-l+1 \ ; k , k ,l+1\right)
\end{equation*}

 \

\vskip .1cm

\noindent $\bullet$ Case $w_1=w_2=w_3=\overline{w}_{3}=0$ and
$\overline{w}_{2}\neq 0$: Using (\ref{m1-28}-\ref{m1-30}) and
(\ref{m1-46}-\ref{m1-48}) we obtain that the Borel subalgebra of
$\so(6)$ annihilates $\overline{w}_{2}$. Hence, it is a highest
weight vector in the irreducible $\cso(6)$-module $F_\mu$, and
considering (\ref{m1-9}-\ref{m1-10}), we have
\begin{equation*}
\mu=\left(-k+1 \ ; l , k ,l+1\right),
\end{equation*}
that is not dominant integral, getting a contradiction.

\vskip .1cm

\noindent $\bullet$ Case
$w_1=w_2=w_3=\overline{w}_{3}=\overline{w}_{2}=0$ and
$\overline{w}_{1}\neq 0$: Using (\ref{m1-25}-\ref{m1-27}) and
(\ref{m1-43}-\ref{m1-45}) we obtain that the Borel subalgebra of
$\so(6)$ annihilates $\overline{w}_{1}$. Hence, it is a highest
weight vector in the irreducible $\cso(6)$-module $F_\mu$, and
considering (\ref{m1-7}-\ref{m1-8}), we have
\begin{equation}\label{m1-mu-3}
\mu=\left(-k \ ; k , l ,l\right), \qquad \hbox{with } 2k\in \ZZ_{>
0}, 2l\in \ZZ_{\geq 0},  \hbox{ and } k-l\in \ZZ_{\geq 0},
\end{equation}
which is dominant integral. In this case, the conditions
(\ref{m1-1}-\ref{m1-51}) reduces to the equation
$E_{-(\ep_2-\ep_3)} \overline{w}_1=0$, that holds for the highest
weight (\ref{m1-mu-3}). Therefore, the vector
\begin{equation}\label{m1-sing-3}
    \vec{m}_1=\big(\xi_{\{2\}^c}-i \xi_{\{1\}^c}\big)\otimes
\overline{w}_1
\end{equation}
is a singular vector of Ind$(F_\mu)$ with $\mu$ as in
(\ref{m1-mu-3}). Now using (\ref{wt-sing-vect}), one can prove
that
\begin{equation*}
wt\ \vec m_1=\left(-k-1 \ ; k+1 , l , l\right)
\end{equation*}
finishing the proof.
\end{proof}

\

\vfill

\pagebreak

\begin{appendices}

\vskip 1cm

\noindent{\bf Link to the folder with the files described below:}

\

https://docs.google.com/leaf?id=0ByKQC9Aglc4YYjQ0NTk3N2YtMm

VhMC00OTY2LWI4MmEtNWVkNmMyOTVkOWY4\&hl=en\_US

\vskip 2cm

%%%%%%%%%%%%%%%%%%%%%%%%%%%%%%%%%%%%%%%%
\section{Notations in the files that use Macaulay2} \lbb{appendix A}
%%%%%%%%%%%%%%%%%%%%%%%%%%%%%%%%%%%%%%%%

\

\vskip .6cm

This appendix contains the explanations of notations used in the
files written for Macaulay2 in order to classify singular vectors
in $CK_6$-induced modules of degree $-1,\dots , -5$. These
notations are the link between the this paper and the files that
use Macaulay2.

As we have seen in (\ref{m-singular}), the possible forms of the
singular vectors are the following:

\

\begin{align}
& \ \vec{m}= \p^2 \ \sum_{|I|= 5} \xi_I\otimes v_{I,2}\ +\ \p \
\sum_{|I|= 3} \xi_I\otimes v_{I,1} +\sum_{|I|= 1} \xi_I\otimes
v_{I,0}, \hbox{ (Degree -5).}\nonumber\\
& \ \vec{m}= \p^2 \ \sum_{|I|= 6} \xi_I\otimes v_{I,2}\ +\  \p \
\sum_{|I|= 4} \xi_I\otimes v_{I,1} +\sum_{|I|= 2} \xi_I\otimes
v_{I,0}, \hbox{ (Degree -4).}\nonumber\\
& \ \vec{m}=   \p \ \sum_{|I|= 5} \xi_I\otimes v_{I,1} +\sum_{|I|=
3} \xi_I\otimes v_{I,0}, \hbox{ (Degree -3).}\nonumber\\
& \  \vec{m}=   \p \ \sum_{|I|= 6} \xi_I\otimes v_{I,1}
+\sum_{|I|= 4} \xi_I\otimes v_{I,0}, \hbox{ (Degree
-2).}\\
& \  \vec{m}=  \sum_{|I|= 5} \xi_I\otimes v_{I,0}, \hbox{ (Degree
-1).}\nonumber
\end{align}

\vskip .6cm

\noindent In order to abbreviate and capture the length of the
elements $\xi_I$ in the summands of the possible singular vectors
$\vec m$, and the degree of $\vec m$, we introduce the following
notation that will be used in the software

\begin{equation}\label{g}
    g_i=\sum_{|I|= i} \xi_I\otimes v_{I,-}
\end{equation}
so they can be  rewritten  as follows:

\

\begin{align}\label{m}
& \ \vec{m_5}= \p^2 \ g_5 \ +\ \p \ g_3 +g_1 , &\hbox{ (Degree -5)}\nonumber\\
& \ \vec{m_4}= \p^2 \ g_6 \ +\  \p \ g_4 +g_2, &\hbox{ (Degree -4)}\nonumber\\
& \ \vec{m_3}=   \p \ g_5 +g_3, &\hbox{ (Degree -3)}\nonumber\\
& \  \vec{m_2}=   \p \ g_6 +g_4, &\hbox{ (Degree -2)}\\
& \  \vec{m_1}=  g_5, &\hbox{ (Degree -1)}\nonumber
\end{align}

\vskip .6cm

We have done a file (or a serie of files in the case of $\vec
m_5,\vec m_3 $ and $\vec m_1$) for each possible singular vector
of type  (\ref{m}). The first part of all the files have the same
structure, and the idea is to impose the equations given in Lemma
\ref{conditions} to each $\vec m_i$. From these equations, we
constructed a matrix by taking the coefficients of these equations
in terms of a natural basis, getting in this way a homogeneous
linear system that is solve in order to get a simplified list of
conditions. Unfortunately we are not expert in Macaulay2 or Maple,
therefore it is not done in the optimal or simpler way.

\vskip .5cm

\centerline{\underline{Description of the inputs:}}

\

\noindent  $\bullet$ Input 1: We define $R0=\QQ[z]/(z^2+1)\simeq
\QQ+i\QQ$. We defined $R0$ because the scalars involved in the
equations of Lemma \ref{conditions} and in the formula of the
$\la$-action belongs to this field.

\

\noindent  $\bullet$ Input 2: We define the polynomial ring $R$
with coefficients in $R0$, in the skew-commutative variables $x_1
,\dots ,x_6$ and the commutative variables $F_{(i,j)} \ (1\leq
i<j\leq 6), \ E\ ,v_I \ (1\leq |I|\leq 6)$. Observe that the
variables $x_i$ correspond to the variables $\xi_i$ in the paper
and $E$ corresponds to the operator $E_{00}$. All the other
variables are the same as in the paper. Note that in this case the
software considers the term $F_{(1,2)}v_3$ as a monomial in the
polynomial ring, not as the element $F_{(1,2)}\in \so (6)$ acting
in $v_3\in F$.

\

\begin{remark} Observe that the command "diff$(x,f)$" in Macaulay2
is the derivative of $f$ {\bf on the right} with respect to $x$.
Since we work with skew-commutative variables $x_i$ and we need to
compute the left derivative $\p_{x_i}$ (see the formula of the
$\la$-action on induced modules). In our case, we have
\begin{equation}\label{diff}
\p_{x_i}(f)=(-1)^{|f|-1} (\hbox{diff}(x_i,f)),
\end{equation}
and
\begin{equation}\label{difff}
    \p_{x_i}\p_{x_j}(f)= -
    (\hbox{diff}(x_i,\hbox{diff}(x_j,f))).
\end{equation}
\end{remark}

\

\noindent  $\bullet$ Input 3: We define $f\_(0)=1_R$,
 $f\_(I)=x_I=\xi_I$ for $1\leq |I|\leq 3$, and
 $fd\_(I)=(f\_(I))^*=\xi_I^*$ for $0\leq |I|\leq 3$. Observe that
 we used "diff" in the definition of $fd$.

\

\noindent  $\bullet$ Input 4: For $1\leq i\leq 6$, we define
$g\_(i)=\sum_{|I|= i} x_I* v\_(I)$ as in (\ref{g}) and (\ref{m}).

\

\noindent  $\bullet$ Inputs 5-10: Now we write the terms used in
the notation introduced in (\ref{eq:ffff}) and (\ref{eq:fd}).
Namely, we define the terms

\vskip .1cm

\begin{align}
&a\_((I),k) :=a(\xi_I,g_k), &ad\_((I),k) :=ad(\xi_I,g_k),\nonumber\\
&b\_((I),k) :=b(\xi_I,g_k), &bd\_((I),k) :=bd(\xi_I,g_k),\nonumber\\
&B\_((I),k) :=B(\xi_I,g_k), &Bd\_((I),k) :=Bd(\xi_I,g_k),\nonumber\\
&C\_((I),k) :=C(\xi_I,g_k), &Cd\_((I),k) :=Cd(\xi_I,g_k),\nonumber
\end{align}

\vskip .2cm

\noindent for all $1\leq k\leq 6$ and $0\leq |I|\leq 3$. Observe
that in order to write the terms that appear in the $\la$-action,
we have to take care of the sign in the derivative by using
(\ref{diff}) and (\ref{difff}).

\

\noindent  $\bullet$ Inputs 11-17: According to Lemma
\ref{conditions}, the conditions (S1)-(S3) on a vector $\vec{m}$,
of degree at most -5, are equivalent to the following list of
equations

\vskip 0.3cm

\noindent $\underline{* \hbox{ For  } |f|=0:}$

\begin{align}
 &0=\mathit{C_0} + \mathit{B_1},&ec\_((0),1) \label{ec-a-1}\nonumber\\
 &0=2\ \mathit{B_2} + {\displaystyle \mathit{a_2}} +
{\displaystyle \mathit{C_1}} ,&ec\_((0),2)\nonumber\\
&0=2\ \mathit{bd_0} - \,i\,\mathit{a_2} +  \,i
\,\mathit{C_1}.&ec\_((0),3) \nonumber
\end{align}

\vskip .1cm

\noindent $\underline{* \hbox{ For  } f=\xi_i:}$

\begin{align}
& 0= 3 \ \mathit{B_2} + 2 \,i\,\mathit{bd_1} + 2
\,i\,\mathit{ad_0} + {\displaystyle
2\, \mathit{C_1}} ,& ec\_((i),1)\nonumber\\
& 0=2\ \mathit{C_0} - \mathit{a_1 }  +
\mathit{B_1}  + 2\ \mathit{bd_0} \,i, & ec\_((i),2)\nonumber\\
& 0= 2 \mathit{a_2} + B_2 ,& ec\_((i),3)\nonumber\\
& 0= 3\ \mathit{Bd_0} - \,i\,\mathit{C_1} + \mathit{bd_1}  -
2\,\mathit{ad_0} , & ec\_((i),4)\nonumber\\
& 0= 2\ \mathit{b_2} +  \mathit{a_1}  +
\mathit{B_1} , & ec\_((i),5)\nonumber\\
& 0= \mathit{b_1} + \mathit{B_0}.& ec\_((i),6)\nonumber
\end{align}

\vskip .1cm

\noindent $\underline{* \hbox{ For  } f=\xi_i\xi_j \ \ (i<j):}$

\begin{align}
& 0= 2\ \mathit{C_0} + 2\ \mathit{Bd_0}\,i + \mathit{B_1}  -
 \,i\,\mathit{ad_0} +  \,i\,
\mathit{bd_1} ,  & ec\_((i,j),1)\nonumber\\
& 0= 2\ \mathit{b_2} + \,i\,\mathit{ad_0} +  \, i\,\mathit{bd_1} +
\mathit{B_1} , & ec\_((i,j),2)\nonumber \\
& 0= \mathit{bd_0} + \mathit{b_1}\,i - \mathit{B_0}\,i.&
ec\_((i,j),3)\nonumber
\end{align}

\vskip .1cm

\noindent $\underline{* \hbox{ For  } f=\xi_i\xi_j\xi_k\ \
(i<j<k):}$

\begin{align}
& 0= \mathit{C_0} -\mathit{Cd_0}\,i, & ec\_((i,j,k),1)\nonumber\\
& 0= \mathit{bd_0} + i \,\mathit{b_0}, & ec\_((i,j,k),2)\nonumber\\
& 0= \mathit{B_1} - \mathit{Bd_1}\,i - \mathit{a_1} +
\mathit{ad_1}\,i, & ec\_((i,j,k),3)\nonumber\\
& 0= \mathit{b_2} -\mathit{bd_2}\,i + \mathit{a_1} -
\mathit{ad_1}\,i, & ec\_((i,j,k),4)\nonumber\\
& 0= \mathit{bd_1} + \mathit{Bd_0} + \mathit{B_0}\,i +
\mathit{b_1}\,i, & ec\_((i,j,k),5)\nonumber\\
& 0= \mathit{ad_0} + \mathit{a_0}\,i - \mathit{Bd_0} -
\mathit{B_0}\,i . & ec\_((i,j,k),6)\nonumber
\end{align}

\vskip .1cm

\noindent $\underline{* \hbox{ For  } f=\alpha_{ij} \hbox{ or }
\beta_{ij}\in B_{\so(6)}\ \ (1\leq i<j\leq 3):}$

\begin{align}
b_1(\alpha_{ij}) & =0, & ecborel\_((i,j),1)\nonumber\\
b_1(\beta_{ij}) & =0, & ecborel\_((i,j),2)\nonumber\\
b_2(\alpha_{ij}) & =0, & ecborel\_((i,j),3)\nonumber\\
b_2(\beta_{ij}) & =0, & ecborel\_((i,j),4)\nonumber\\
b_0(\alpha_{ij}) & =0, & ecborel\_((i,j),5)\nonumber\\
b_0(\beta_{ij}) & =0, & ecborel\_((i,j),6)\nonumber
\end{align}
The right column of the previous list of conditions contains the
name that is used in the Macaulay file of $\vec m_5$ for each
equation. Observe that for each vector $\vec m_i$ the equations
are implemented in a different way, taking care of the elements
$g_k$. Namely, if we work with $\vec m_4=\p^2 g_6+\p g_4+g_2$,
then equation $ec\_((0),3)$ is written in Macaulay file as
\begin{equation}\label{ex}
2\ \mathit{bd\_((0),2)} - \,z\,\mathit{a\_((0),6)} + \,z
\,\mathit{C\_((0),4)}=0,
\end{equation}
where $z$ corresponds to the complex number $i$. And for $\vec
m_5=\p^2 g_5+\p g_3+g_1$, then equation $ec\_((0),3)$ is written
in the corresponding Macaulay file as
\begin{equation}\label{ex}
ec\_((0),3)=2\ \mathit{bd\_((0),1)} - \,z\,\mathit{a\_((0),5)} +
\,z \,\mathit{C\_((0),3)}=0.
\end{equation}

Not all the equations are non-trivial for the different $\vec
m_i$, since the length of the monomial $\xi_I$ may be greater than
6. For example, in (\ref{ex}), the length of $\xi_I$ is 6, but
this equation is trivial when it is implemented for $\vec m_3$
since the length of $\xi_I$ is 7. In the following table we
indicate which equations appear for the different $\vec m_i$ and
we give the length of $\xi_I$ that is present in each case.
Therefore, the name and number of the equations is modified for
the file of each $\vec m_i$.

\

\begin{tabular}{||c||c|c|c|c|c||}
  \hline
      & $\vec m_5$ & $\vec m_4$& $\vec m_3$ & $\vec m_2$ & $\vec m_1$
      \\\hline
      & & & & & \\
  \underline{$\bullet\  |f|=0:$}$\qquad\qquad\qquad\qquad\qquad\qquad$ &  &  &  &  &  \\
   & & & & & \\
  $\mathit{C_0} + \mathit{B_1}$& 3 & 4 & 5 & 6 & - \\
   $2\ \mathit{B_2} + {\displaystyle \mathit{a_2}} +
{\displaystyle \mathit{C_1}}$ & 5 & 6 & - & - & - \\
 $ 2\ \mathit{bd_0} - \,i\,\mathit{a_2} +  \,i
\,\mathit{C_1}$ & 5 & 6 & - & - & -\\
  & & & & & \\
   \underline{$\bullet \ |f|=1:$}$\qquad\qquad\qquad\qquad\qquad\qquad$ &  &  &  &  &  \\
  & & & & & \\
  $3 \ \mathit{B_2} + 2 \,i\,\mathit{bd_1} + 2
\,i\,\mathit{ad_0} + {\displaystyle
2\, \mathit{C_1}} $ & 6 & - & - & - & - \\
$2\ \mathit{C_0} - \mathit{a_1 }  +
\mathit{B_1}  + 2\ \mathit{bd_0} \,i$ & 4 & 5 & 6 & - & - \\
 $2\mathit{a_2} + B_2 $ & 6 & - & - & - & - \\
$  3\ \mathit{Bd_0} - \,i\,\mathit{C_1} + \mathit{bd_1}  -
2\,\mathit{ad_0}  $ & 6 & - & - & - & - \\
$ 2\ \mathit{b_2} +  \mathit{a_1}  +
\mathit{B_1} $ & 4 & 5 & 6 & - & - \\
$ \mathit{b_1} + \mathit{B_0} $ & 2 & 3 & 4 & 5 & 6 \\
  & & & & & \\
  \underline{$\bullet\ |f|=2:$} $\qquad\qquad\qquad\qquad\qquad\qquad$ &  &  &  &  &  \\
  & & & & & \\
 $ 2\ \mathit{C_0} + 2\ \mathit{Bd_0}\,i + \mathit{B_1}  -
 \,i\,\mathit{ad_0} +  \,i\,
\mathit{bd_1} $ & 5 & 6 & - & - & - \\
$  2\ \mathit{b_2} + \,i\,\mathit{ad_0} +  \, i\,\mathit{bd_1} +
\mathit{B_1} $ & 5 & 6 & - & - & - \\
$ \mathit{bd_0} + \mathit{b_1}\,i - \mathit{B_0}\,i $ & 3 & 4 & 5 & 6 & - \\
  & & & & & \\
   \underline{$\bullet \ |f|=3:$} $\qquad\qquad\qquad\qquad\qquad\qquad$ &  &  &  &  &  \\
  & & & & & \\
$  \mathit{C_0} -\mathit{Cd_0}\,i $ & 6 & - & - & - & - \\
$   \mathit{bd_0} + i \,\mathit{b_0} $  & 2 & 3 & 4 & 5 & 6 \\
$  \mathit{B_1} - \mathit{Bd_1}\,i - \mathit{a_1} +
\mathit{ad_1}\,i $ & 6 & - & - & - & - \\
$  \mathit{b_2} -\mathit{bd_2}\,i + \mathit{a_1} -
\mathit{ad_1}\,i $ & 6 & - & - & - & - \\
$  \mathit{bd_1} + \mathit{Bd_0} + \mathit{B_0}\,i +
\mathit{b_1}\,i $ & 4 & 5 & 6 & - & - \\
$  \mathit{ad_0} + \mathit{a_0}\,i - \mathit{Bd_0} -
\mathit{B_0}\,i $ & 4 & 5 & 6 & - & - \\
  & & & & & \\
\underline{$\bullet \ f\in $Borel:}$\qquad\qquad\qquad\qquad\qquad\qquad$ &  &  &  &  &  \\
  & & & & & \\
$  b_2 $ & 5 & 6 & - & - & - \\
 $ b_1 $ & 3 & 4 & 5 & 6 & - \\
 $ b_0 $ & 1 & 2 & 3 & 4 & 5 \\ \hline
\end{tabular}

\

\noindent  $\bullet$ Input 18-21: We denote by $A\_((I),k)$ a one
column matrix whose entries are the coefficients in the monomials
$x_J$ of the equation $ec\_((I),k)$. We should impose that the
equation of each entry must be zero.

\

\noindent  $\bullet$ Input 22: We denote by $M\_((I),k)$ a one
column matrix whose entries are the coefficients in the monomials
$x_J$ of the equation $ecborel\_((I),k)$. We should impose that
the equation of each entry must be zero.

\

\

\noindent  $\bullet$ Input 23: The previously defined matrices
$A\_((I),k)$ and  $M\_((I),k)$ are one column matrices whose
entries are R0-linear combinations of the monomials $v\_(I), $ $
F\_(i,j)*v\_(I)$ and $E*v\_(I)$. Each entry must be zero, for that
reason we define the lists zvari=list$\{v\_(I), \ F\_(i,j),\ E\}$
and wvari=list$\{v\_(I), \ F\_(i,j)*v\_(I),\ E*v\_(I)\}$, in this
order, with the auxiliary lists avari1, avari2 and avari3.

\

\noindent  $\bullet$ Input 28-31: We take the transpose of
$A\_((I),k)$ getting a one row matrix. Then for each entry in this
one row matrix, we produce a column formed by the coefficients of
this entry with respect to the variables in wvari, obtaining in
this way a matrix with coefficients in R0 whose transpose is
called $D\_((I),k)$. If we consider wvari as a one column matrix,
then we have $A\_((I),k)=D\_((I),k)*$wvari and it must be zero.
Therefore we obtained a homogeneous linear system that must be
solved.

\

\noindent  $\bullet$ Input 32: With the same procedure, using the
matrices $M\_((I),k)$, we define the matrices $N\_((I),k)$ that
complete the linear system.

\

\noindent  $\bullet$ Input 33: The matrices $D\_((I),k)$ and
$N\_((I),k)$ are put together into one matrix that is called X
whose coefficients are in R0. So, we need to solve the homogeneous
liner system associated to X.

\

\

Observe that with this procedure, we consider the elements
$F\_(i,j)*v\_(I)$ as a monomial  in the ring $R$, not as
 an element  in $\so (6)$ acting in $v\_(I)$. Since the software (at least from
  our knowledge) does not work with Lie theory, we first solve the linear system,
 and then we impose the Lie setting by hand. The description of the
inputs that we gave is essentially the structure of all the
Macaulay files associated to the vectors $\vec m_5, \vec m_4$ and
$\vec m_2$. The files associated to $\vec m_3$ and $\vec m_1$ have
a modification: before the definition of the matrices "D" and "M"
all the variables $F\_(i,j)$ are written as linear combinations of
the more natural basis of $\so (6)$ given by the $H_i$ and
$E_\alpha$. Therefore the list of monomials in wvari is written in
terms of them.

Now, we describe in details the list of files associated to each
$\vec m_i$.

\

\

\centerline{\underline{Files associated to $\vec m_5$}}

\

\noindent \underline{$\bullet$ File "m5-macaulay-1"}

\

With the list of inputs previously described, we get a $1952
\times 544$ matrix $X$ of rank $540$ (see inputs 33-40). This
matrix $X$ is constructed by joining together the list of matrices
$l0,l1,\dots ,l4$. In order to reduce the size of the matrix, we
study the rank of these matrices and we found that the $992\times
544$ matrix, called $Y25$, formed with the matices $l0,l1,l2,l4$
also has rank $540$. Unfortunately, the software Macaulay2 can
solve a linear system if the matrix is over $\ZZ_p, \mathbb{R}$ or
$\CC$, and it must be a non-singular square matrix in the cases $
\mathbb{R}$ or $\CC$. Therefore, we exported the matrix $Y25$ and
we used Maple, see the file "m5-maple-1", to find the row-reduced
echelon matrix of $Y25$, that is called $C$ in that file.

\

\noindent \underline{$\bullet$ File "m5-macaulay-2"}

\

If we try to copy the matrix $C$ in the file "m5-macaulay-1" the
software run out of memory. Therefore, we continue the work in
this NEW Macaulay file "m5-macaulay-2". Now, we describe the
inputs in details:

\

\noindent $\ast$ Input 1-7: The rings $R0$ and $R$, and the list
of variables $wvari$ are copied from the file "m5-macaulay-1". We
need $wvari$ because, in  input 19, we reconstruct the (reduced)
equations as linear combinations of the monomials
$v\_(I),F\_(i,j)*v\_(I)$ and $E*v\_(I)$.

\vskip .1cm

\noindent $\ast$ Input 8-15: The matrix $C$ that is produced in
the file "m5-maple-1", which is  the row-reduced echelon matrix of
$Y25$, is introduced in this NEW Macaulay file "m5-macaulay-2"
divided in several parts, called $X1,\dots, X7$. These parts are
put together to reconstruct the matrix $C$ and it is called $X11$
(input 15). Observe that $Y25$ was a $992\times 544$ matrix of
rank $540$. For this reason, we copied the first 542 rows of $C$
(the row-reduced echelon matrix of $Y25$). Therefore $X11$ is a
$542\times 544$ matrix with zero in the last two  rows.

\vskip .1cm

\noindent $\ast$ Input 16-19: We obtain a reduced (and equivalent)
list of equations in a one column matrix $X28=X27*wvari$ (whose
size is $542\times 1$), where $X27$ is $X11$ viewed with entries
in the ring $R$. Each entry must be zero.

At the end of this list of equations, we observe the following
conditions:

\begin{align}\label{m5-cond}
& 0= \
 v_1+v_{1,3,4,5,6}                    \ \
      & 0= \  v_2-v_{1,3,4,5,6}\ i                    \nonumber  \\
      & 0= \  v_3+v_{1,2,3,5,6}                    \ \
      & 0= \  v_4-v_{1,2,3,5,6}\ i                    \nonumber  \\
      & 0= \  v_5+v_{1,2,3,4,5}                    \ \
      & 0= \  v_6-v_{1,2,3,4,5}\ i                    \nonumber  \\
      & 0= \  v_{1,2,3}-v_{1,2,3,5,6}\ i             \ \
      & 0= \  v_{1,2,4}-v_{1,2,3,5,6}               \nonumber  \\
      & 0= \  v_{1,2,5}-v_{1,2,3,4,5}\ i             \ \
      & 0= \  v_{1,2,6}-v_{1,2,3,4,5}               \nonumber  \\
      & 0= \  v_{1,3,4}-v_{1,3,4,5,6}\ i             \ \
      & 0= \  v_{1,3,5}+v_{2,4,6}\ i                  \nonumber  \\
      & 0= \  v_{1,3,6}+v_{2,4,6}                  \ \
      & 0= \  v_{1,4,5}+v_{2,4,6}                     \\
      & 0= \  v_{1,4,6}-v_{2,4,6}\ i                 \ \
      & 0= \  v_{1,5,6}-v_{1,3,4,5,6}\ i              \nonumber  \\
      & 0= \  v_{2,3,4}-v_{1,3,4,5,6}              \ \
      & 0= \  v_{2,3,5}+v_{2,4,6}                   \nonumber  \\
      & 0= \  v_{2,3,6}-v_{2,4,6}\ i                 \ \
      & 0= \  v_{2,4,5}-v_{2,4,6}\ i                  \nonumber  \\
      & 0= \  v_{2,5,6}-v_{1,3,4,5,6}              \ \
      & 0= \  v_{3,4,5}-v_{1,2,3,4,5}\ i              \nonumber  \\
      & 0= \  v_{3,4,6}-v_{1,2,3,4,5}              \ \
      & 0= \  v_{3,5,6}-v_{1,2,3,5,6}\ i              \nonumber  \\
      & 0= \  v_{4,5,6}-v_{1,2,3,5,6}              \ \
      & 0= \  v_{2,3,4,5,6}+v_{1,3,4,5,6}\ i          \nonumber  \\
      & 0= \  v_{1,2,4,5,6}+v_{1,2,3,5,6}\ i         \ \
      & 0= \  v_{1,2,3,4,6}+v_{1,2,3,4,5}\ i           \nonumber
\end{align}

\

\noindent In order to simplify the 540 equations, we need to
impose conditions (\ref{m5-cond}). Observe that all the vectors
$v_I$ can be written in terms of the set
$\{v\_1,v\_3,v\_5,v\_(1,3,5)\}$. This is done in the following
inputs.

\vskip .1cm

\noindent $\ast$ Input 20-21:  We define a ring $P$ that is
isomorphic to $R$. In this case, $P$ is  the polynomial ring with
coefficients in $R0$, in the skew-commutative variables $t\_1
,\dots ,t\_6$ and the commutative variables $h\_i, e_{(i,j)},
em_{(i,j)}, me_{(i,j)}, mem_{(i,j)} \ (1\leq i<j\leq 3), \ E0\
,u_I \ (1\leq |I|\leq 6)$. The idea is to replace the basis
$F_{(i,j)}\in \so (6)$ by the basis  given by  $H_i$ and
$E_\alpha$. We are using the following notation, for $1\leq
i<j\leq 3$:

\begin{align}\label{mem}
e\_(i,j) & = E_{\ep_i-\ep_j} \nonumber \\
em\_(i,j) & = E_{\ep_i+\ep_j} \\
me\_(i,j) & = E_{-(\ep_i-\ep_j)} \nonumber \\
mem\_(i,j) & = E_{-(\ep_i+\ep_j)} \nonumber
\end{align}

\vskip .1cm

\noindent $\ast$ Input 22: We define a map $Q:R\to P$, that impose
conditions (\ref{m5-cond}) and change the basis in $\so (6)$ using
the notation (\ref{mem}). The definition of $Q$ is the following:
%
%\vskip .1cm
%
%
\begin{align*}
& \hbox{Input 22: } \ Qvari=\{ x\_i=>t\_i,
\\ &
F\_(1,2)=>-z*h\_1,F\_(3,4)=>-z*h\_2,F\_(5,6)=>-z*h\_3,
\\ &
F\_(2*i-1,2*j-1)=>(e\_(i,j)+em\_(i,j)+me\_(i,j)+mem\_(i,j))/4), \\
& F\_(2*i,2*j)=>(e\_(i,j)-em\_(i,j)+me\_(i,j)-mem\_(i,j))/4), \\ &
F\_(2*i-1,2*j)=>-z*(e\_(i,j)-em\_(i,j)-me\_(i,j)+mem\_(i,j))/4),
\\ &
F\_(2*i,2*j-1)=>-z*(-e\_(i,j)-em\_(i,j)+me\_(i,j)+mem\_(i,j))/4),
\\ & E=>E0, \\ &
v\_1=>u\_1,v\_2=>-z*u\_1,v\_3=>u\_3,
\\ &
v\_4=>-z*u\_3,v\_5=>u\_5,v\_6=>-z*u\_5,
\\ &
v\_(i,j)=>u\_(i,j)),
\\ &
v\_(1,2,3)=>-z*u\_3,v\_(1,2,4)=>-u\_3,v\_(1,2,5)=>-z*u\_5,
\\ &
v\_(1,2,6)=>-u\_5,v\_(1,3,4)=>-z*u\_1,v\_(1,3,5)=>u\_(1,3,5),
\\ &
v\_(1,3,6)=>-z*u\_(1,3,5),
v\_(1,4,5)=>-z*u\_(1,3,5),v\_(1,4,6)=>-u\_(1,3,5),
\\ &
v\_(1,5,6)=>-z*u\_1,
v\_(2,3,4)=>-u\_1,v\_(2,3,5)=>-z*u\_(1,3,5),
\\ &
v\_(2,3,6)=>-u\_(1,3,5),
v\_(2,4,5)=>-u\_(1,3,5),v\_(2,4,6)=>z*u\_(1,3,5),
\\ &
v\_(2,5,6)=>-u\_1, v\_(3,4,5)=>-z*u\_5,v\_(3,4,6)=>-u\_5,
\\ &
v\_(3,5,6)=>-z*u\_3,v\_(4,5,6)=>-u\_3,
\\
& v\_(i,j,k,l)=>u\_(i,j,k,l)))),
\\ &
v\_(2,3,4,5,6)=>z*u\_1,v\_(1,3,4,5,6)=>-u\_1,
v\_(1,2,4,5,6)=>z*u\_3,
\\ &
v\_(1,2,3,5,6)=>-u\_3,v\_(1,2,3,4,6)=>z*u\_5,
v\_(1,2,3,4,5)=>-u\_5,
\\ &
v\_(1,2,3,4,5,6)=>u\_(1,2,3,4,5,6)\},
\\ &
Q=map(P,R,Qvari);
\end{align*}

\vskip .1cm

\noindent $\ast$ Input 23-27: We define a new 'wvari', which is
the list of variables that will appear in the equations. More
precisely, wvari=list$\{u\_(I), \ h\_k*u\_(I),\ e\_(i,j)*u\_(I),\
em\_(i,j)*u\_(I),\ me\_(i,j)*u\_(I),\ mem\_(i,j)*u\_(I),\
E0*u\_(I)\}$, where $u\_(I)$ is restricted in this case to the set
$\{u\_1,u\_3,u\_5,u\_(1,3,5)\}$.

\vskip .1cm

\noindent $\ast$ Input 28-31: We apply the map $Q$ to the
equations in the matrix $X28$, and then we obtain a $542\times 68$
matrix, called $X29$, given by the coefficients in the monomials
of 'wvari' that appear in the equations of $Q(X28)$. The rank of
X29 is 64. The matrix $X29$ is exported in order to reduce the
linear system.

\vskip .1cm

\noindent $\ast$ Input 33: We use Maple to reduce the matrix
$X29$. The matrix $C$ that is produced in the file "m5-maple-2",
which is the row-reduced echelon matrix of $X29$, is introduced in
this input and it is called $X30$. Observe that $X29$ was a
$542\times 68$ matrix of rank $64$. For this reason, we copied the
first 70 rows of $C$ (the row-reduced echelon matrix of $X29$).
Therefore $X30$ is a $70\times 68$ matrix with zero in the last
six rows.

\vskip .1cm

\noindent $\ast$ Input 34-38: We obtain a reduced (and equivalent)
list of equations in a one column matrix $X34=X33*wvari$ (whose
size is $70\times 1$), where $X33$ is $X30$ viewed with entries in
the ring $P$. Each entry must be zero. This final list of 64
simplified and equivalent equations is copied in the proof of
Lemma \ref{m5}.

\

\

\centerline{\underline{Files associated to $\vec m_4$}}

\

\noindent \underline{$\bullet$ File "m4-macaulay"}

\

With the list of inputs previously described, except that we do
not need to impose Borel equations (hence the matrices "M" and "N"
are not needed), we get a $1104 \times 527$ matrix $X$ of rank
$527$. Therefore, there is no non-trivial solution of this linear
system, proving that there is no singular vector of degree -4.

\

\

\centerline{\underline{Files associated to $\vec m_3$}}

\

\noindent \underline{$\bullet$ File "m3-macaulay-1"}

\

With the list of inputs previously described, we get a $694 \times
442$ matrix $X$ of rank $397$.  This matrix $X$ is exported to a
file and using Maple, see the file "m3-maple-1", we obtain the
row-reduced echelon matrix of $X$, that is called $C$.  This
matrix $C$ is introduced as the matrix $X11$ in the Macaulay file
(see input 43-44), to reconstruct the (reduced) equations as
linear combinations of the monomials $v\_(I),F\_(i,j)*v\_(I)$ and
$E*v\_(I)$. In fact, the matrix $X11$ is $400\times 442$ because
we removed the last zero rows of the row-reduced echelon matrix,
therefore it has zero in the last three  rows.

\vskip .1cm

\noindent $\ast$ Input 45-48: We obtain a reduced (and equivalent)
list of equations in a one column matrix $X14=X12*wvari$ (whose
size is $400\times 1$), where $X12$ is $X11$ viewed with entries
in the ring $R$. Each entry must be zero.

\

At the end of this list of equations, we observe the following
conditions:

\begin{align}\label{m3-cond}
                   & 0 =  \ v_{1,2,3}-v_{4,5,6}\ i
       & 0 =  \  v_{1,2,4}+v_{3,5,6}\ i
      \nonumber \\ & 0 =  \  v_{1,2,5}-v_{3,4,6}\ i
        & 0 =  \  v_{1,2,6}+v_{3,4,5}\ i
      \nonumber \\ & 0 =  \  v_{1,3,4}-v_{2,5,6}\ i
        & 0 =  \  v_{1,3,5}+v_{2,4,6}\ i
      \nonumber \\ & 0 =  \  v_{1,3,6}-v_{2,4,5}\ i
        & 0 =  \  v_{1,4,5}-v_{2,3,6}\ i
      \nonumber \\ & 0 =  \  v_{1,4,6}+v_{2,3,5}\ i
        & 0 =  \  v_{1,5,6}-v_{2,3,4}\ i
                \\ & 0 =  \  v_{2,3,4,5,6}
        & 0 =  \  v_{1,3,4,5,6} \hskip .88cm \
      \nonumber \\ & 0 =  \  v_{1,2,4,5,6}
        & 0 =  \  v_{1,2,3,5,6} \hskip .88cm \
      \nonumber \\ & 0 =  \  v_{1,2,3,4,6}
        & 0 =  \  v_{1,2,3,4,5} \hskip .88cm \ \nonumber
 \end{align}

\

\noindent Observe that (\ref{m3-cond}) can be written as
\begin{equation}\label{m3-cond-2}
    v_I=0 \hbox{ if } |I|=5, \ \ \hbox{ and }\
    v_{\{a,b,c\}}=(-1)^{a+b+c}\ i\ v_{\{a,b,c\}^c} \ \hbox{ for }
    a<b<c.
\end{equation}

\

\noindent \underline{$\bullet$ File "m3-macaulay-2"}

\

Now, we have to impose the identities (\ref{m3-cond}) to reduce
the number of variables. Using (\ref{m3-cond-2}), everything can
be written in terms of
\begin{equation*}\label{1111}
    v_{1,j,k} \ \ \hbox{ with } \ 2\leq j<k\leq 6.
\end{equation*}
Unfortunately, the result is not enough to obtain in a clear way
the possible highest weight vectors. For example, after the
reduction and some extra computations it is possible to see that
\begin{equation}\label{hhh}
(v_{1,3,6}-v_{1,4,5})+\ i\ (v_{1,3,5}+v_{1,4,6})
\end{equation}
is  annihilated by the Borel subalgebra. Hence, it is necessary to
impose (\ref{m3-cond}) and  make a change of variables. We
produced an auxiliary file where we imposed (\ref{m3-cond}), and
after the analysis of the results, we found that the following
change of variable is convenient:
\begin{align}\label{ui}
u_1 & = v_{1,2,3}- i\ v_{1,2,4} \nonumber \\
u_2 & = v_{1,2,3}+ i\ v_{1,2,4} \nonumber \\
u_3 & = v_{1,2,5}- i\ v_{1,2,6} \nonumber \\
u_4 & = v_{1,2,6}+ i\ v_{1,2,6} \nonumber \\
u_5 & = v_{1,3,4}- \ v_{1,5,6}  \\
u_6 & = v_{1,3,4}+ \ v_{1,5,6} \nonumber \\
u_7 & = v_{1,3,5}-v_{1,4,6}+ i\ (v_{1,3,6}+v_{1,4,5}) \nonumber \\
u_8 & = v_{1,3,5}+v_{1,4,6}- i\ (v_{1,3,6}-v_{1,4,5}) \nonumber \\
u_9 & = v_{1,3,5}-v_{1,4,6}- i\ (v_{1,3,6}+v_{1,4,5}) \nonumber \\
u_{10} & = v_{1,3,5}+v_{1,4,6}+ i\ (v_{1,3,6}-v_{1,4,5}) \nonumber
\end{align}
or equivalently
\begin{align}\label{ch-var}
v_{1,2,3} & = \frac{1}{2}(u_1+u_2) = i\ v_{4,5,6} \nonumber \\
v_{1,2,4} & = \frac{i}{2}(u_1-u_2) = -i \ v_{3,5,6} \nonumber \\
v_{1,2,5} & = \frac{1}{2}(u_3+u_4) = i \ v_{3,4,6} \nonumber \\
v_{1,2,6} & = \frac{i}{2}(u_3-u_4) = -i \ v_{3,4,5} \nonumber \\
v_{1,3,4} & = \frac{1}{2}(u_5+u_6) = i \ v_{2,5,6} \\
v_{1,5,6} & = \frac{-1}{2}(u_5-u_6) = i \ v_{2,3,4} \nonumber \\
v_{1,3,5} & = \frac{1}{4}(u_7+u_8+u_9+u_{10}) = - i \ v_{2,4,6} \nonumber \\
v_{1,4,6} & = \frac{1}{4}(-u_7+u_8-u_9+u_{10})= -i \ v_{2,3,5} \nonumber \\
v_{1,3,6} & = \frac{i}{4}(-u_7+u_8+u_9-u_{10})= i \ v_{2,4,5} \nonumber \\
v_{1,4,5} & = \frac{i}{4}(-u_7-u_8+u_9+u_{10})= i \ v_{2,3,6}
\nonumber
\end{align}

We also need to replace the  basis $F_{(i,j)}\in \so (6)$ by the
basis given by  $H_i$ and $E_\alpha$. We are using the  notation
introduced in (\ref{mem}). The identities (\ref{m3-cond}), the
change of variables (\ref{ch-var}) and the new basis of $\so (6)$
are implemented with the definition of a ring $P$ that is
isomorphic to $R$ and a map $Q:R\to P$. More precisely,

\vskip .1cm

\noindent $\ast$ Input 1-22: They are the same inputs in the file
"m3-macaulay-1".

\vskip .1cm

\noindent $\ast$ Input 23: We define a ring $P$ as the polynomial
ring with coefficients in $R0$, in the skew-commutative variables
$t\_1 ,\dots ,t\_6$ and the commutative variables $u_i \ (1\leq
i\leq 10), h\_i, e_{(i,j)}, em_{(i,j)}, me_{(i,j)}, mem_{(i,j)} \
(1\leq i<j\leq 3), \ E0$.

\vskip .1cm

\noindent $\ast$ Input 24: We define a map $Q:R\to P$, that impose
conditions (\ref{m3-cond}), the change of variables (\ref{ch-var})
and change the basis in $\so (6)$ using the notation (\ref{mem}).
The definition of $Q$ is the following:
\begin{align*}
& Qvari=\{ x_i=>t_i,
\\
& F\_(1,2)=>-z*h_1,F\_(3,4)=>-z*h_2,F\_(5,6)=>-z*h_3, \\
& F\_(2*i-1,2*j-1)=>(e\_(i,j)+em\_(i,j)+me\_(i,j)+mem\_(i,j))/4,\\
& F\_(2*i,2*j)=>(e\_(i,j)-em\_(i,j)+me\_(i,j)-mem\_(i,j))/4, \\
& F\_(2*i-1,2*j)=>-z*(e\_(i,j)-em\_(i,j)-me\_(i,j)+mem\_(i,j))/4,
\\
& F\_(2*i,2*j-1)=>-z*(-e\_(i,j)-em\_(i,j)+me\_(i,j)+mem\_(i,j))/4,
\\
& E=>E0, \\
&  v_i=>u_i, \\
& v\_(i,j)=>u\_(i,j),
\end{align*}
\begin{align*}
& v\_(1,2,3)=>(u_1+u_2)/2,v\_(1,2,4)=>z*(u_1-u_2)/2,\\
& v\_(1,2,5)=>(u_3+u_4)/2, v\_(1,2,6)=>z*(u_3-u_4)/2,
\\
& v\_(1,3,4)=>(u_5+u_6)/2,v\_(1,5,6)=>-(u_5-u_6)/2,
\\
& v\_(1,3,5)=>(u_7+u_8+u_9+u_{10})/4,\\
&
v\_(1,3,6)=>z*(-u_7+u_8+u_9-u_{10})/4,\\
& v\_(1,4,5)=>z*(-u_7-u_8+u_9+u_{10})/4,\\
&
v\_(1,4,6)=>(-u_7+u_8-u_9+u_{10})/4,\\
& v\_(2,3,4)=>z*(u_5-u_6)/2, \\
& v\_(2,3,5)=>z*(-u_7+u_8-u_9+u_{10})/4,\\
& v\_(2,3,6)=>(-u_7-u_8+u_9+u_{10})/4,\\
& v\_(2,4,5)=>(-u_7+u_8+u_9-u_{10})/4,\\
& v\_(2,4,6)=>z*(u_7+u_8+u_9+u_{10})/4,\\
& v\_(2,5,6)=>-z*(u_5+u_6)/2,v\_(3,4,5)=>-(u_3-u_4)/2,\\
& v\_(3,4,6)=>-z*(u_3+u_4)/2,v\_(3,5,6)=>-(u_1-u_2)/2,\\
& v\_(4,5,6)=>-z*(u_1+u_2)/2, \\
& v\_(i,j,k,l)=>u\_(i,j,k,l),\\
& v\_(1..i-1|i+1..6)=>0_P, \\
& v\_(1,2,3,4,5,6)=>u\_(1,2,3,4,5,6)\},\\
& Q=map(P,R,Qvari).
\end{align*}

\vskip .1cm

\noindent $\ast$ Input 25-41: With the same  list of inputs as in
the file "m3-macaulay-1", but applying the map $Q$, we get a $694
\times 170$ matrix $X$ of rank $125$. This matrix $X$ is
constructed by joining together the list of matrices $l0,l1,\dots
,l4$.

\vskip .1cm

\noindent $\ast$ Input 42-47: In order to reduce the size of the
matrix, we studied the rank of these matrices and we found that
the $354\times 170$ matrix, called $X2$, formed with the matices
$l0,l1,l2,l4$ also has rank $125$.  We exported the matrix $X2$
and we used Maple, see the file "m3-maple-2", to find the
row-reduced echelon matrix of $X2$, that is called $C$ in that
file.

\vskip .1cm

\noindent $\ast$ Input 48-49: The matrix $C$ that is produced in
the file "m3-maple-2", which is  the row-reduced echelon matrix of
$X2$, is introduced in this  file  and it is called $X11$. Observe
that $X2$ was a $354\times 170$ matrix of rank $125$. For this
reason, we copied the first 127 rows of $C$ (the row-reduced
echelon matrix of $X2$). Therefore $X11$ is a $127\times 170$
matrix with zero in the last two  rows.

\vskip .1cm

\noindent $\ast$ Input 50-55: We obtain a reduced (and equivalent)
list of equations in a one column matrix $X14=X12*wvari$ (whose
size is $127\times 1$), where $X12$ is $X11$ viewed with entries
in the ring $P$. Each entry must be zero.

\vskip .1cm

\noindent $\ast$ Input 56-60: In order to simplify the list of
equations, we define $Z\_i=$row $i$ of $X14$, and then we consider
the following list of linear combinations of these rows:

\vskip .1cm

\begin{align*}
& \hbox{for i in 0..124 do } Z\_i=X14\_(i,0);\\
& i57:\\
&
X25=\{Z\_0,Z\_1-Z\_0,Z\_2+Z\_0,Z\_3,Z\_4+Z\_3,Z\_5-Z\_3,Z\_6,Z\_7+Z\_6,\\
&
Z\_8-Z\_6,Z\_9+Z\_11,Z\_10+Z\_11,Z\_11,Z\_12,Z\_13-Z\_12,Z\_14+Z\_13,\\
&
Z\_15+Z\_17,Z\_16+Z\_15,Z\_17+Z\_16,Z\_18+Z\_19-Z\_20,Z\_18-Z\_19+Z\_20,\\
&
-Z\_18+Z\_19+Z\_20,Z\_21,Z\_22,Z\_23+Z\_22,Z\_24+Z\_25+Z\_26,\\
& \hbox{for i in 25..29 list }
Z\_i,\\
&
Z\_30-Z\_0,Z\_31-Z\_3,Z\_32-Z\_6,Z\_33+Z\_11,Z\_34-Z\_14,Z\_35-Z\_15,\\
&
\hbox{for i in 36..124 list } Z\_i\};
\end{align*}

\vskip .1cm

\noindent obtaining an equivalent list of equations given by the
rows of the $127\times 1$ matrix $X29$ (all of them must be zero).
These equations are copied  in the proof of Lemma \ref{5.5}
together with the final analysis of them, see
(\ref{C-1}-\ref{C-125}).

\

\

\centerline{\underline{Files associated to $\vec m_2$}}

\

\noindent \underline{$\bullet$ File "m2-macaulay"}

\

With the list of inputs previously described, we get a $268 \times
272$ matrix $X$ of rank $192$. This matrix $X$ is exported to a
file and using Maple (see the file "m2-maple.mws") we obtain the
row-reduced echelon matrix of $X$, that is called $X11$. In fact,
the matrix $X11$ is $195\times 272$ because we removed the last
zero rows of the row-reduced echelon matrix. This matrix $X11$ is
introduced in the Macaulay file "m2-macaulay" as the input 43. In
order to reconstruct the reduced system of equations as linear
combinations of the monomials $v\_(I),F\_(i,j)*v\_(I)$ and
$E*v\_(I)$ we multiply X11*wvari, obtaining a one column matrix,
called X14, with the list of equations that must be zero (see
inputs 45-52). This matrix $X14$ is exported into a latex-pdf file
 "m2-ecuations.pdf", and the analysis of these equations is done
 in the paper (see the proof of the lemma \ref{m2} corresponding to $\vec
 m_2$).

\

\centerline{\underline{Files associated to $\vec m_1$}}

\

\noindent \underline{$\bullet$ File "m1-macaulay"}

\

Since the singular vectors found in \cite{BKL} for $K_6$ are also
singular vectors for $CK_6$,  using (B42-B43) in \cite{BKL}, we
have that it is convenient to introduce the following
notation:
\begin{align}\label{vect-m1}
\vec{m}_1 & =\sum_{i=1}^6 \xi_{\{i\}^c}\otimes v_{\{i\}^c}\\
&
=\sum_{l=1}^3\bigg[\big(\xi_{\{2l\}^c}+i\xi_{\{2l-1\}^c}\big)\otimes
w_l + \big(\xi_{\{2l\}^c}-i \xi_{\{2l-1\}^c}\big)\otimes
\overline{w}_l\bigg]\nonumber
\end{align}
that is, for $1\leq l\leq 3$
\begin{equation}\label{v-w}
    v_{\{2l\}^c}  =w_l+\overline{w}_l,\qquad
    v_{\{2l-1\}^c} =i(w_l-\overline{w}_l)
\end{equation}
or equivalently, for $1\leq l\leq 3$
\begin{equation}\label{w-v}
    w_{l}  =\frac{1}{2}(v_{\{2l\}^c}-i\ v_{\{2l-1\}^c}),\qquad
    \overline{w}_{l} =\frac{1}{2}(v_{\{2l\}^c}+i\ v_{\{2l-1\}^c}).
\end{equation}

Now, with the usual list of inputs previously described, we have
 to impose the identities (\ref{v-w}) to change the variables. We
also need to replace the  basis $F_{(i,j)}\in \so (6)$ by the
basis given by  $H_i$ and $E_\alpha$. We are using the  notation
introduced in (\ref{mem}). The change of variables (\ref{v-w}) and
the new basis of $\so (6)$ are implemented with the definition of
a ring $P$ that is isomorphic to $R$ and a map $Q:R\to P$. More
precisely,

\vskip .1cm

\noindent $\ast$ Input 1-19: They are the usual inputs, for
example as  in the file "m3-macaulay-1".

\vskip .1cm

\noindent $\ast$ Input 20-21: We define a ring $P$ as the
polynomial ring with coefficients in $R0$, in the skew-commutative
variables $t\_1 ,\dots ,t\_6$ and the commutative variables $u_i \
(1\leq i\leq 10),$   $
omega_1,omega_2,omega_3,domega_1,domega_2,domega_3,$   $ h\_i,
e_{(i,j)}, em_{(i,j)},$  \linebreak  $ me_{(i,j)}, mem_{(i,j)} \ (1\leq i<j\leq
3), \ E0$.

\vskip .1cm

\noindent $\ast$ Input 22: We define a map $Q:R\to P$, that change
of variables (\ref{v-w}) and change the basis in $\so (6)$ using
the notation (\ref{mem}). The definition of $Q$ is the following:
\begin{align*}
& Qvari=\{ x\_i=>t\_i, \\
& F\_(1,2)=>-z*h\_1,F\_(3,4)=>-z*h\_2,F\_(5,6)=>-z*h\_3, \\
& F\_(2*i-1,2*j-1)=>(e\_(i,j)+em\_(i,j)+me\_(i,j)+mem\_(i,j))/4), \\
& F\_(2*i,2*j)=>(e\_(i,j)-em\_(i,j)+me\_(i,j)-mem\_(i,j))/4),  \\
& F\_(2*i-1,2*j)=>-z*(e\_(i,j)-em\_(i,j)-me\_(i,j)+mem\_(i,j))/4),
\\
& F\_(2*i,2*j-1)=>-z*(-e\_(i,j)-em\_(i,j)+me\_(i,j)+mem\_(i,j))/4),\\
& E=>E0, \\
& v\_i=>u\_i, \\
& v\_(i,j)=>u\_(i,j),  \\
& v\_(i,j,k)=>u\_(i,j,k), \\
& v\_(i,j,k,l)=>u\_(i,j,k,l), 
\end{align*}
\begin{align*}
& v\_(2,3,4,5,6)=>z*(omega\_1-domega\_1), \qquad \qquad \qquad \qquad \qquad \qquad \ \\
& v\_(1,2,4,5,6)=>z*(omega\_2-domega\_2), \\
& v\_(1,2,3,4,6)=>z*(omega\_3-domega\_3), \\
& v\_(1,3,4,5,6)=>(omega\_1+domega\_1),\\
& v\_(1,2,3,5,6)=>(omega\_2+domega\_2), \\
& v\_(1,2,3,4,5)=>(omega\_3+domega\_3), \\
& v\_(1,2,3,4,5,6)=>u\_(1,2,3,4,5,6)\}, \\
& Q=map(P,R,Qvari).
\end{align*}

\vskip .1cm

\noindent $\ast$ Input 23-37: With the usual   list of inputs, but
applying the map $Q$ to the equations, we get a $62 \times 102$
matrix $X$ of rank $51$. We exported the matrix $X$ and we used
Maple, see the file "m1-maple", to find the row-reduced echelon
matrix of $X$, that is called $C$ in that file.

\vskip .1cm

\noindent $\ast$ Input 39: The matrix $C$ that is produced in the
file "m1-maple", which is  the row-reduced echelon matrix of $X$,
is introduced in this  file  and it is called $X11$. Observe that
$X$ was a $62\times 102$ matrix of rank $51$.  Therefore $X11$ is
a $62\times 102$ matrix with zero in the last 11  rows.

\vskip .1cm

\noindent $\ast$ Input 40-45: We obtain a reduced (and equivalent)
list of equations in a one column matrix $X14=X12*wvari$ (whose
size is $62\times 1$), where $X12$ is $X11$ viewed with entries in
the ring $P$. Each entry must be zero.

\vskip .1cm

These equations are copied in the proof of Lemma \ref{m1}, in the
equations (\ref{m1-1}-\ref{m1-51}), and the final analysis of them
is done in that proof.

 \

\end{appendices}

\vskip 1cm

\

\noindent{\bf Acknowledgment.} C. Boyallian and J. Liberati were
supported in
  part by grants of
Conicet, Foncyt-ANPCyT and  Secyt-UNC  (Argentina). V. Kac was
supported in part by an NSF grant.

%%%%%%%%%%%%%%%%%%%%%%%%%%%%%%%%%%%%%%%%%%%%%%%%%%%%%%%%%%%%%%%%%%%%%%%%%
%\bibliographystyle{amsplain}

%%%%%%%%%%%%%%%%%%%%%%%%%%%%%%%%%%%%%%%%%%%%%%%%%%%%%%%%%%%%%%%%%%%%%%%%%


\begin{thebibliography}{11}
%%%%%%%%%%%%%%%%%%%%%%%%%%%%%%%%%%%%%%%%%%%%%%%%%%%%%%%%%%%%%%%%%%%%%%%%%






\bibitem [1]{BKLR} C. Boyallian; V.G. Kac; J. Liberati and A. Rudakov,
{\em Representations of simple finite Lie conformal superalgebras of type $W$ and $S$.}
J. Math. Phys. {\bf 47} (2006), no. 4, 043513.

\bibitem [2]{BKL} C. Boyallian; V.G. Kac and J. Liberati,
{\em Irreducible modules over finite simple Lie conformal
superalgebras of type K}, J. Math. Phys. {\bf 51} (2010) 063507,
37pp.


\bibitem [3]{CK}  S. Cheng  and  V. G. Kac, {\em Conformal
    modules}, Asian J. Math. {\bf 1} (1997), 181-193.  Erratum: 2
  (1998), 153-156.


\bibitem [4]{CK2}  S. Cheng  and  V. G. Kac, {\em  A
new $N=6$ superconformal algebra}, Comm. Math. Phys. {\bf 186}
(1997), no. 1, 219-231.

\bibitem [5]{CK3}  S. Cheng  and  V. G. Kac, {\em  Structure of some $\ZZ$-graded Lie superalgebras of vector fields}, Transf. Groups. {\bf 4} (1999),
219-272.

\bibitem [6]{CL}
S. Cheng  and N. Lam, {\textit{Finite conformal modules over
$N=2,3,4$ superconformal algebras.}}, Journal of Math. Phys. {\bf
42}(2001), 906-933.


\bibitem [7]{FK}  D. Fattori and V. G. Kac,  {\em Classification of
finite simple Lie conformal superalgebras.} J. Algebra {\bf
258} (2002), no. 1, 23--59.


\bibitem[8]{K1} V.~G.~Kac,
{\textit{Vertex algebras for beginners}}, University Lecture
Series, 10. American Mathematical Society, Providence, RI, 1996.
Second edition 1998.

\bibitem[9]{KR1} V. G. Kac and A. Rudakov,  {\em
Representations of the exceptional Lie superalgebra $E(3,6)$. I.
Degeneracy conditions}, Transform. Groups {\bf 7} (2002), no. 1,
67--86.

\bibitem[10]{KR2} V. G. Kac and A. Rudakov,  {\em Complexes of modules
over exceptional Lie superalgebras $E(3,8)$ and $E(5,10)$}, IMRN
{\bf 19} (2002), 1007--1025.

\bibitem[11]{Knapp} A. W. Knapp,
{\em Lie groups beyond an introduction}, Progress in
Mathematics 140. Birkh\"auser, 1996.



%%%%%%%%%%%%%%%%%%%%%%%%%%%%%%%%%%%%%%%%%%%%%%%%%%%%%%%%%%%%%%%%%%%%%%%%%
\end{thebibliography}
\end{document}